\documentclass[opre,nonblindrev]{informs3_hide}

\DoubleSpacedXI % Made default 4/4/2014 at request
\usepackage[english]{babel}
\usepackage[autostyle, english = american]{csquotes}
\MakeOuterQuote{"}

\usepackage[margin=1in]{geometry}

\usepackage{algpseudocode}
\usepackage[ruled, vlined, norelsize]{algorithm2e}
\usepackage{amsmath,amssymb}
\usepackage{mathtools}
\usepackage{bbm} % for \mathbbm{1} symbol

\usepackage[inline]{enumitem}

\usepackage{makecell}
\usepackage{booktabs}
\usepackage{hyperref}
\usepackage{cleveref}
\crefname{subsection}{subsection}{subsections}

\newcommand{\cP}{\mathcal{P}}
\newcommand{\ssw}{w^*}
\newcommand{\ssx}{x^*}
\newcommand{\sss}{s^*}
\newcommand{\set}[1]{\{#1\}}
\newcommand{\bpgraph}[1][G]{#1 = (U, V; E)}

\newcommand{\euler}{\mathbf{\mathsf{e}}} %%% Euler's constant
\newcommand{\patience}{\theta}
\newcommand{\pat}{\patience}
\newcommand{\opat}{\overline{\pat}}

\newcommand{\eps}{\varepsilon}
\newcommand{\bI}{\mathbbm{1}}
\newcommand{\bE}{\mathbb{E}}

\newcommand{\ALG}{\mathsf{ALG}}
\newcommand{\OPT}{\mathsf{OPT}}
\newcommand{\LPOPT}{\mathsf{OPT}_{\mathsf{LP}}}
\newcommand{\OPTP}{\mathsf{OPT}_{\mathsf{LPP}}}

\newcommand{\LRG}{I_{\mathrm{LARGE}}}
\newcommand{\SML}{I_{\mathrm{SMALL}}}

\makeatletter
\newcommand{\bigger}{\bBigg@{4}}
\newcommand{\biggerr}{\bBigg@{5}}
\makeatother

\usepackage{natbib,bbm,multirow,multicol}
\bibpunct[, ]{(}{)}{,}{a}{}{,}%
\def\bibfont{\small}%
\TheoremsNumberedThrough     % Preferred (Theorem 1, Lemma 1, Theorem 2)
\ECRepeatTheorems

\EquationsNumberedThrough    % Default: (1), (2), ...
\MANUSCRIPTNO{OPRE-2021-06-371.R2} % When the article is logged in and DOI assigned to it,
\begin{document}
%%%%%%%%%%%%%%%%

%\setcounter{page}{1}
%\DOI{10.1287/xxxx.0000.0000}%


% \RUNTITLE{}

\TITLE{Online Matching Frameworks under Stochastic Rewards, Product Ranking, and Unknown Patience}

\ARTICLEAUTHORS{
\AUTHOR{Brian Brubach}
\AFF{Computer Science Department, Wellesley College, \EMAIL{bb100@wellesley.edu}, \URL{}}
\AUTHOR{Nathaniel Grammel}
\AFF{Department of Computer Science, University of Maryland, College Park, \EMAIL{ngrammel@umd.edu}, \URL{}}
\AUTHOR{Will Ma}
\AFF{Graduate School of Business, Columbia University, New York, \EMAIL{wm2428@gsb.columbia.edu}, \URL{}}
\AUTHOR{Aravind Srinivasan}
\AFF{Department of Computer Science, University of Maryland, College Park, \EMAIL{asriniv1@umd.edu}, \URL{}}
}

\ABSTRACT{

\SingleSpacedXI We study generalizations of online bipartite matching in which each arriving vertex (customer) views a ranked list of offline vertices (products) and matches to (purchases) the first one they deem acceptable. The number of products that the customer has patience to view can be stochastic and dependent on the products seen. We develop a framework that views the interaction with each customer as an abstract resource consumption process, and derive new results for these online matching problems under the adversarial, non-stationary, and IID arrival models, assuming we can (approximately) solve the product ranking problem for each single customer. To that end, we show new results for product ranking under two cascade-click models: an optimal algorithm when each item has its own hazard rate for making the customer depart, and a 1/2-approximate algorithm when the customer has a general item-independent patience distribution. We also present a constant-factor 0.027-approximate algorithm in a new model where items are not initially available and arrive over time. We complement these positive results by presenting three additional negative results relating to these problems.
% Finally, we present three negative results of interest: one formalizing the notion of a stochasticity gap exhibited by existing approaches to this problem, an example showing the analysis of SimpleGreedy in existing work to be tight, and another one for the single-customer problem in which any constant-factor approximation is impossible when compared to a benchmark that knows the realization of the patience in advance.

}

%\KEYWORDS{}


\maketitle
%%%%%%%%%%%%%%%%%%%%%%%%%%%%%%%%%%%%%%%%%%%%%%%%%%%%%%%%%%%%%%%%%%%%%%

%\include{response-to-reviewers}

%\clearpage

\section{Introduction}

Online matching is a fundamental problem in e-commerce and online advertising, introduced in the seminal work of \citet{kvv}. While offline matching has a long history in economics and computer science, online matching has exploded in popularity with the ubiquity of the internet and the emergence of online marketplaces. 
A common scenario in e-commerce is the online sale of unique goods due to the ability to reach niche markets via the internet (e.g., eBay); typical products include rare books, trading cards, art, crafts, and memorabilia. 
We will use this as a motivating example to describe our setting. However, the settings we study can also model job search/hiring, crowdsourcing, online advertising, ride-sharing, and other online-matching problems.

In classical online bipartite matching, we start with a known set of \textit{offline vertices} that may represent items for sale or ads to be allocated. Then, an unknown sequence of \textit{online vertices} arrive, which may represent customers, users, or visitors to a webpage.
These online vertices or customers arrive one-by-one, and the decision to match each customer or not (and if so, to which item) must be made irrevocably before the next customer is revealed. In the original formulation, the online vertices are chosen fully adversarially, although models that assume they are drawn from probability distributions have since been studied \citep{bib:Feldman,alaei2012online}.

Many generalizations of online matching have also been proposed, including stochastic rewards and weighted graphs.
Under stochastic rewards, there can be repeated interactions with a customer (recommending an item, and if they do not accept, recommending another item, etc.) before the next one arrives, as we describe subsequently.
Our paper's goal is to provide a framework that decouples the repeated-interaction problem for a single customer from the overall allocation problem over multiple customers, leading to new results and unification of old ones.
Moreover, motivated by product ranking, we derive new results for the repeated-interaction problem with a single customer, including the extension in which the horizon for these interactions is unknown or stochastic.

\textbf{Description of stochastic rewards model, with patience.}
In the \textit{stochastic rewards} model, each edge exists independently according to a known probability; this probability is revealed upon the arrival of its incident online vertex.
This is motivated by online platforms in which only a probabilistic prediction of whether a customer will buy an item is known at the time they arrive.
The algorithm can "probe" edges incident to an online vertex, or equivalently recommend the customer an item, after which if they accept, then the item is sold committedly.
If they otherwise reject, then under the basic stochastic rewards model \citep{mehta2012online} there is no opportunity to offer another item; this is known as the customer having a \textit{patience} of 1.

Other papers \citep{BansalLPCure,AGMesa,brubach2017attenuate} have more generally allowed the customer to have any deterministic patience $\pat$.
This can be interpreted as a \textit{product ranking} problem where $\pat$ different items are listed on a page, and the customer will view them in order, stopping once they see an acceptable item, or reaching the end of the page.
The product ranking problem where $\pat$ is deterministic can be efficiently solved using dynamic programming \citep{purohit2019hiring}.
More generally, if $\pat$ is random, but drawn from a known distribution, then the customer may probabilistically depart after seeing any undesirable item; this is called the \textit{cascade-click} model of product ranking.
We will derive new results for the cascade-click model.

\textbf{Description of edge weights and stochastic arrival models.} In an orthogonal
generalization of online bipartite matching, edges between items and customers
may have a \textit{weight}, which is the reward collected when that edge is matched. This
can represent, e.g., the price at which that item is sold to the customer. When
edges can take on different possible weights, parametric competitive ratios are
known \citep{ma2020algorithms}, but a competitive ratio that is an absolute constant is impossible in the original adversarial arrival model
\citep[see][]{mehtaBook}. Therefore, many papers have focused on the relaxed models of \emph{stochastic arrivals}
or \emph{vertex weights} instead, each of which circumvents this impossibility.

In the stochastic arrival models, the total number of online vertices $T$ is
known, and each online vertex $t=1,\ldots,T$ has a \textit{type}\footnote{This includes everything that is known about the customer at the time of their arrival, including purchase probabilities, patience distribution, edge weights, etc.} drawn independently from a
known distribution. Generally we allow distributions to be \textit{non-stationary} and
vary with $t$, although we also consider the IID special case where these
distributions are identical. Stochastic arrival models are motivated by settings
with sufficient data to estimate the distribution over types.
% ; and although the
% assumption that $T$ is known seems restrictive, in this paper we also consider the setting where $T$ is random.
% we derive a \textit{further impossibility result}
% that shows an absolute-constant competitive ratio is again impossible if $T$ is
% unknown, even under the IID case.
Meanwhile, in the model with vertex weights,
all edges incident to any offline vertex $u$ must have the same weight. This is
motivated by each offline item having its own fixed price that is identical across
customers.

\subsection{Our Contributions}

% Through our decoupling framework, we derive new results for online bipartite
% matching when customers can have stochastic rewards/patience (see \Cref{subsec1} below).
% \red{
We develop a decoupling framework, which we describe in greater detail in
  \Cref{subsec1} below, wherein we first study a simpler, single-customer
  version of various stochastic matching problems, and then use the algorithms
  for these problems to inform decisions during online customer arrivals. This approach allows us to derive new results for online bipartite matching with stochastic rewards and, in many cases, stochastic patience as well.
  % }

We consider both vertex weights and general edge weights in combination with the
adversarial, non-stationary, and IID arrival models. Since our framework
requires the repeated-interaction problems to be solvable for a single customer,
we also make advancements on this front (see \Cref{subsec2}); namely, improving
algorithms for the \textit{cascade-click} model of product ranking, and deriving new
algorithms in a model where \textit{items are arriving} over time. Finally, we derive
several negative results of interest (see \Cref{subsec3}).
% , including one that shows the
% competitive ratio to be 0 in any problem where $\theta$ or $T$ is unknown, if
% comparing against a benchmark that knows the realization of the patience/horizon
% length in advance.

\subsubsection{Framework that decouples online matching from single-customer problems.} \label{subsec1}
First we build a framework that takes as input a subroutine for solving the single-customer problem, and outputs an algorithm for the overall multi-customer online matching problem, under the aforementioned variants.
The competitive ratios guaranteed by our framework are explained in \Cref{tab:matchingresults}, and we would like to highlight a key distinction in our approach.
Existing analyses of stochastic rewards \citep{BansalLPCure,mehta2012online,AGMesa,brubach2017attenuate} all use an LP that is \textit{specific} to the stochastic rewards matching process, which exhibits a stochasticity gap (see \Cref{sec:stochgap}).
By contrast, our framework uses an \textit{abstract} LP, in which:
\begin{itemize}
\item There is a variable $x_v(\pi)$ for each of the (exponentially-many) policies $\pi$ that could be used for interacting with a single customer of type $v$;
\item Each such policy $\pi$ ends up matching each available offline vertex $u$ with probability $p_{uv}(\pi)$;
\item There is a single set of constraints tying together the customers over time, which enforce that each offline vertex $u$ is matched at most once in expectation.
\end{itemize}
Our framework abstracts away the details of the stochastic rewards matching process, deterministic vs.\ stochastic patience, etc., and holds as long as the policies represent different consumption processes\footnote{
Similar ideas have appeared in \citet{cheung2022inventory}, who consider abstract "actions" that have different immediate rewards and different consumption distributions over resources.  However, their focus is on learning these distributions.
} that use up the offline vertices $u$ \textit{independently} over time according to known probabilities.

Our results, however, are predicated on the existence of a subroutine that can (approximately) solve the repeated-interaction problems for each customer.
We will call such a subroutine \textit{$\kappa$-approximate} if given any weights $w_{uv}$, it finds a policy $\pi$ whose immediate expected reward $\sum_u w_{uv}p_{uv}(\pi)$ is at least $\kappa$ times the maximum possible immediate expected reward over all policies, for some $\kappa\in[0,1]$.
Equipped with a $\kappa$-approximate subroutine, our framework provides:
\begin{enumerate}
\item A $\kappa/2$-competitive algorithm for vertex weights and adversarial arrivals \textbf{(\Cref{sec:adversarial})};
\item A $\kappa/2$-competitive algorithm for edge weights and non-stationary arrivals \textbf{(\Cref{sec:edgeprophet})};
\item A $(1-1/\euler)\kappa$-competitive algorithm for edge weights and IID arrivals \textbf{(\Cref{sec:prophetiid})}; and
\item A $(1-1/\euler)\kappa$-competitive algorithm for vertex weights and non-stationary arrivals \textbf{(\Cref{sec:prophetvw})}.
\end{enumerate}
The value $\kappa=1$ is possible when $\theta$ is deterministic \citep{purohit2019hiring}.  We derive new results below showing that $\kappa=1$ is also possible when $\theta$ follows an (item-dependent) hazard rate model, and that $\kappa=1/2$ is possible when $\theta$ follows any (item-independent) distribution.
This, in conjunction with our framework, justifies all of the results in \Cref{tab:matchingresults}.

\begin{table}
\TABLE
{Landscape of Online Matching Results \label{tab:matchingresults}}
{\begin{tabular}{ l l  l  l }
\multicolumn{1}{l}{} \\
\multicolumn{1}{l}{} \\
%%%%% ADVERSARIAL
\textbf{\large Adversarial} 	& \thead{\textbf{Unweighted}} 		& \thead{\textbf{Vertex-weighted}} 	& \thead{\textbf{Edge-weighted}} \\ 
\toprule
\textbf{Non-stochastic} 		& \makecell{$0.632$ (tight)\\ \cite{kvv}} 	& \makecell{$0.632$ (tight)\\ \cite{aggarwalVertex}} 		& \makecell{[must be weight-dependent]\\ \citet{ma2020algorithms}} \\ 
\midrule
\textbf{Stochastic Rewards} 	& \makecell{$0.5$\\ \cite{mehta2012online}} 	&  \makecell{$? \to \textbf{0.5}$}		& \makecell{[must be weight-dependent]\\ \citet{ma2020algorithms}} \\ 
\midrule
\makecell{\textbf{Deterministic Patience/}\\ \textbf{Hazard Rate Model}} 		& \makecell{$? \to \textbf{0.5}$}	& \makecell{$? \to \textbf{0.5}$}		& \makecell{--} \\
\midrule
%\textbf{Constant Hazard Rate} 		& \makecell{$? \to \textbf{0.5}$}	& \makecell{$? \to \textbf{0.5}$}		& \makecell{--} \\
%\midrule
\textbf{Stochastic Patience} 		& \makecell{$? \to \textbf{0.25}$}	& \makecell{$? \to \textbf{0.25}$}		& \makecell{--} \\
\bottomrule
\multicolumn{1}{l}{} \\
\multicolumn{1}{l}{} \\
%%%%% NON-STATIONARY
\textbf{\large Non-stationary} 
& \thead{\textbf{Unweighted}} 		& \thead{\textbf{Vertex-weighted}} 	& \thead{\textbf{Edge-weighted}} \\ 
\toprule
\textbf{Non-stochastic} & \makecell{$0.632$\\ \cite{alaei2012online}} & \makecell{$0.632$\\ \cite{alaei2012online}} & \makecell{$0.5$\\ \cite{alaei2012online}} \\ 
\midrule
\makecell{\textbf{Deterministic Patience/}\\ \textbf{Hazard Rate Model}} 	& \makecell{$? \to \textbf{0.632}$}  & \makecell{$? \to \textbf{0.632}$}	& \makecell{$? \to \textbf{0.5}$}   \\
\midrule
%\textbf{Constant Hazard Rate} 		& \makecell{$? \to \textbf{0.632}$}	& \makecell{$? \to \textbf{0.632}$}		& \makecell{$? \to \textbf{0.5}$} \\
%\midrule
\textbf{Stochastic Patience} 		& \makecell{$? \to \textbf{0.316}$}	& \makecell{$? \to \textbf{0.316}$}		& \makecell{$? \to \textbf{0.25}$} \\
\bottomrule
\multicolumn{1}{l}{} \\
\multicolumn{1}{l}{} \\
%%%%% KNOWN IID
\textbf{\large Known IID} 
& \thead{\textbf{Unweighted}} 		& \thead{\textbf{Vertex-weighted}} 	& \thead{\textbf{Edge-weighted}} \\ 
\toprule
\textbf{Non-stochastic} 	& \makecell{$0.729$\\ \cite{Brubach2016arxiv}}& \makecell{$0.729$\\ \cite{Brubach2016arxiv}}	& \makecell{$0.705$\\ \cite{Brubach2016arxiv}} \\ 
\midrule
\textbf{Stochastic Rewards} 	& \makecell{$0.632$\\ \cite{Brubach2016arxiv}}		& \makecell{$0.632$\\ \cite{Brubach2016arxiv}}		& \makecell{$0.632$\\ \cite{Brubach2016arxiv}} \\ 
\midrule
\makecell{\textbf{Deterministic Patience/}\\ \textbf{Hazard Rate Model}} 	& \makecell{$0.46$ $\to \textbf{0.632}$\\ \cite{brubach2017attenuate}} & \makecell{$0.46$ $\to \textbf{0.632}$\\ \cite{brubach2017attenuate}}	&   \makecell{$0.46$ $\to \textbf{0.632}$\\ \cite{brubach2017attenuate} } \\
\midrule
%\textbf{Constant Hazard Rate} 		& \makecell{$? \to \textbf{0.632}$}	& \makecell{$? \to \textbf{0.632}$}		& \makecell{$? \to \textbf{0.632}$} \\
%\midrule
\textbf{Stochastic Patience} 		& \makecell{$? \to \textbf{0.316}$}	& \makecell{$? \to \textbf{0.316}$}		& \makecell{$? \to \textbf{0.316}$} \\
\bottomrule
\end{tabular}}
{Landscape of online matching results grouped by arrival model, form of edge weights, and including the unknown patience models we introduce: the (item-dependent) hazard rate model, and the arbitrary (item-independent) stochastic patience model. \textbf{Bold} results with arrows show the improvements from this paper, with question marks denoting problems where no prior bound was known.
% If a result follows immediately from the work of a paper, we cite that paper even if the specific result was not mentioned in the paper itself.
}
\end{table}

\subsubsection{New $\kappa$-approximate subroutines for single-customer problems.} \label{subsec2}
As discussed above, it is important for our framework to have $\kappa$-approximate subroutines for the repeated-interaction problems with a single customer.
We make the following advancements on this front:
\begin{enumerate}
\item A 1-approximate (optimal) subroutine, in the model where each item $i$ has a known \textit{hazard rate} $r_i$ and, if seen by the customer and undesired, causes the customer to depart with probability (w.p.)~$r_i$;
\item A 1/2-approximate subroutine, in the model where the customer has an arbitrary known patience distribution (and the probabilities of departing do not depend on the items seen);
\item A 0.027-approximate subroutine, in a new model where the customer has a deterministic patience, but the items are arriving over time according to Bernoulli processes.
\end{enumerate}
The first two models can be motivated by product ranking in e-commerce.
A special case of the first model \textbf{(\Cref{sec:starconsthazard})} is where $r_i$ is equal to some $r$ for all $i$, which represents a patience distribution with \textit{constant} hazard rate $r$, i.e.\ a customer who departs w.p.~$r$ after each position regardless of the item seen.
Meanwhile, our 1/2-approximation for the second model \textbf{(\Cref{sec:arbpatience})} improves the state-of-the-art $1/e$-approximation from \citet{chen2020revenue} for this cascade-click model of product ranking.
Their result also only holds in the special case of increasing hazard rate, while we extend it to general distributions by formulating and rounding a new LP relaxation\footnote{
However, we acknowledge that their $1/\euler$-approximation holds against a stronger benchmark that knows the patience in advance. This is only possible under some special cases of the patience distribution: as we show in \Cref{sec:unknown_pat}, such a result is \textit{impossible} for the general patience distributions we consider, so our LP relaxation (necessarily) does not know the patience in advance.
} for this single-customer problem.
We note that general patience distributions are well-motivated in applications; see e.g.\ \citet{aveklouris2021matching}, who study a matching model where items are also arriving over time.
On that note, our result for the third model \textbf{(\Cref{sec:item_arrivals})}, when plugged into our frameworks, provides constant-factor guarantees in a related model where items (representing contractors in an online labor platform) may not be present at the beginning and need to arrive online after each customer (to acknowledge they can perform the customer's task), and the customer has to then also accept that contractor.
We contrast this new model with other online platform matching models in \Cref{sec:relatedWork}.

\subsubsection{Negative results.} \label{subsec3}
Finally, our work presents three important negative results.
\begin{enumerate}
\item We formalize the notion of a stochasticity gap for LP-based approaches to these problems, and construct a stochastic bipartite graph in which even the offline maximum matching has expected size at most 0.544 times the value of the LP relaxation \textbf{(\Cref{sec:stochgap})}.  This means that the competitive ratio from the existing LP-based approaches cannot be better than 0.544, while our framework yields a $1-1/\euler\approx0.632$-competitive algorithm.
\item We show that the simple family of greedy algorithms introduced in \citet{mehta2012online} cannot be better than 1/2-competitive \textbf{(\Cref{sec:simple_greedy})}.
\item
We show that when offering items to a single customer with random patience, if one compares to a benchmark that knows the realization of the patience in advance, then any constant-factor approximation is impossible \textbf{(\Cref{sec:unknown_pat})}.
Importantly, our counterexample holds even if the customer can be repeatedly offered the same item, which is identical to having an unknown number of opportunities to make a \textit{single} sale (since the customer will buy at most one item).
This is similar in spirit to the negative result derived in \citet{alijani2020predict}.
%Through a well-known equivalence \citep{correa2019pricing}, the same example shows that \textit{any constant-factor competitive ratio is impossible} in an online accept/reject problem with IID values, \textit{if the total number of values to be drawn is unknown}.
\end{enumerate}

\section{Further Related Work} \label{sec:relatedWork}

\textbf{Online matching with stochastic rewards.}
% Online matching was introduced in~\citet{kvv} for unweighted graphs with adversarial
% arrivals. The authors gave a $(1-1/\euler)$-competitive algorithm (called
% Ranking) and proved this was the best possible.
% % Since it is not possible to
% % extend this to the edge-weighted setting while achieving a competitive ratio,
% Subsequent work considered variants that either relax the adversarial arrival
% setting or make added assumptions about the weights. Other than adversarial,
% common arrival models are random order, known IID, and known non-stationary.
% % (also called prophet inequality matching).
% For all but adversarial arrivals,
% arbitrary edge weights may be considered; however, in some cases, improved
% guarantees can still be achieved under vertex weights or uniform weights
% (i.e.~for unweighted graphs).
Online matching represents a large literature, which has been surveyed in \citet{mehtaBook}.  We will describe the portion of this literature that focuses on stochastic rewards, where edges only match probabilistically upon being probed.  This problem has been studied under both adversarial and stochatic arrival models, as well as different variants depending on the assumptions about edge weights/patience.
% For any combination of arrival model and weight assumption, we may consider
% deterministic edges (as in the classical matching problem) or stochastic edges. In the
% stochastic edges setting, online vertices may be assumed to have a patience of
% one (the "stochastic rewards" setting), or we may consider arbitrary patience
% values on each online vertex.

Online matching with stochastic edges was introduced in~\citet{BansalLPCure}\footnote{Their paper focuses on the \textit{offline} matching with stochastic edges problem, which we do not consider  in this literature review.} as stochastic matching with
timeouts (patience), where the authors showed a ratio of $0.12$ for known IID
arrivals and arbitrary edge weights.
This was later improved to $0.46$ in \citet{brubach2017attenuate}\footnote{The techniques in~\citet{brubach2017attenuate} also involved solving a star graph problem with a black box. However, that work first solved an LP for a bipartite graph, and then used a black box probing algorithm to essentially round and probe the LP solution on the induced star graphs of arriving vertices. This differs from our work which uses algorithms for stochastic matching on star graphs as black boxes to solve a more sophisticated LP, then use that LP solution to guide the online algorithm.}
and to 0.51 in \citet{fata2019multi} for some cases.
We improve these results by establishing a competitive ratio of $1/2$ for non-stationary arrivals and $1-1/\euler$ for IID arrivals.
We note that
\citet{borodin2021prophet} concurrently prove these results, differing in three ways: i) they allow for more general constraints on which edges can be probed, beyond a simple patience constraint (although, they do not consider stochastic patience); ii) they compare against a more powerful offline benchmark that can switch back-and-forth between probing different online vertices; iii) they show that $1-1/\euler$ holds in the more general model of non-identical independent draws arriving in a uniformly random order.
The same authors have also studied online matching with stochastic edges under the "secretary" model of random-order arrival \citep[see][]{borodin2021secretary}.

For adversarial arrivals, most work has focused on the unweighted case,
initially studied by~\citet{mehta2012online} in the special case where patience $\pat_v$ equals $1$
for all $v$. Under the further restriction of uniform vanishing edge
probabilities, they showed that a competitive ratio of $0.53$ is possible. This
was extended to a ratio of $0.534$ for unequal, but still vanishingly small
probabilities~\citep{mehtaonline}. These results were also recently improved to $0.576$ and
$0.572$ respectively by \citet{huang2020online} and then to $0.596$ for both models by \citet{goyalUdwani};\@
however, all these results focus on the case of vanishingly small probabilities,
do not consider patience values greater than 1, and do not consider vertex
weights. For arbitrary edge probabilities, general deterministic patience values, and vertex
weights, our guarantee of 0.5 is the best-known. 
%\red{[ NPG: Is this misleading?
%  For adversarial arrivals and vertex weights, we get $0.5$, but for general
%  patience we get another factor of $0.5$, so all of these together is really
%  $0.25$, not $0.5$. ]}\wnote{I added the qualifier "deterministic".} 
There is also a hardness result in~\citet{mehta2012online} which shows that
no algorithm for stochastic rewards with adversarial arrivals can achieve a
competitive ratio greater than $0.62$. This quantity is strictly less than
$1-1/\euler$, although we argue that this difference is artificially caused by
the stochasticity gap, as we explain in \Cref{sec:stochgap}.

\citet{GNR} study another model of stochastic rewards, in which when a vertex $v$ (viewed as a customer) arrives online, an online algorithm chooses a \emph{set} $S$ of potential matches for $v$ (viewed as an offering of products to the customer). Each customer (online vertex) has a \emph{general choice model} which specifies the probability of the customer purchasing each item when offered each possible set of product assortments $S$. We contrast this model in more detail in \Cref{sec:relation_to_assortment}, but note that in this setting, a set of potential matches is chosen \emph{all at once} rather than probed sequentially, with the outcome being determined by full set $S$ (the offered product assortment).

\textbf{Large starting capacities.} We do not study how our guarantees improve if there are at least $k$ copies of every offline vertex, although we believe our frameworks could be expanded to do so.  The state-of-the-art for these $k$-dependent guarantees in online matching can be found for: adversarial arrivals \citep{ma2020algorithms}, unweighted adversarial arrivals \citep{kalyanasundaram2000optimal}, non-stationary arrivals with vertex weights \citep{alaei2012online}, general non-stationary arrivals \citep{jiang2022tight}, and IID arrivals \citep{ma2021dynamic}.

\textbf{Cascade-click models in product ranking.}
We turn our literature review to papers that study the repeated-interaction/product ranking problems for a single customer.
% We now explain the correspondence between our patience-constrained problem for a single customer and cascade-click models in the product ranking literature.
% We interpret each offline vertex $i$ as a product, $p_i$ as the probability of $i$ being purchased when viewed, and $w_i$ as the fixed revenue garnered from selling product $i$.
% The decision is a ranking of the products.
% Then, a single customer arrives, viewing the products in order and purchasing each viewed product $i$ according to an independent coin flip of probability $p_i$.
% We emphasize that the customer does not look forward in the ranking, purchasing the first product whose coin realizes positively; this is also described as the customer following a "satisficing" choice model \citep[see][]{gallego2017attention}.
% The customer exits after either purchasing a product, or running out of patience to view more products, with the patience drawn independently from a known distribution.
% \citet{chen2020revenue} establish a $1/\euler$-approximation for this problem under the additional assumption of the patience distribution having increasing hazard rate.
% Our result from \Cref{sec:arbpatience} is directly a 1/2-approximation for this problem without any assumptions on the patience distribution, based on the correspondence described above.
Our result from \Cref{sec:starconsthazard} shows how to optimally solve this problem under constant hazard rate, a special case of interest in \citet{chen2020revenue}.
Our result in \Cref{sec:arbpatience} improves their guarantee and holds for general patience distributions.
We should note that our results do not directly apply to more general cascade-click models \citep[see][]{kempe2008cascade} where the probability of the customer running out of patience depends on the specific item shown, but we believe that our simple LP-based technique in \Cref{sec:arbpatience} could be useful for these generalized models.
Other generalized ranking problems involving choice models are studied in \citet{derakhshan2018product}.

%\textbf{Sequential assortment problem.}
In the related \textit{sequential assortment} problem, multiple products can be shown to the customer at a time.
The customer chooses between them according to a Multi-Nomial Logit (MNL) choice model (instead of independent click probabilities), and alternatively the customer could choose the option of viewing the next assortment, never to return.
There is a constraint that the same product cannot be shown in different assortments.
This problem is typically studied when the number of stages is deterministic and known (see \citet{feldman2019improved} and the references therein); however, the \textit{stage-dependent} coefficients in \citet{feldman2019improved} can be used to capture our notion of a stochastic patience.
Nonetheless, the PTAS derived in \citet{feldman2019improved} does not subsume our 1/2-approximation because: in the sequential assortment problem there is no constraint on the number of products offered at once, hence it does not capture our problem; also, in a PTAS, to get a $(1-\eps)$-approximation the runtime needs to be exponential in $1/\eps$, whereas our LP-based technique has polynomial runtime independent of any error parameter.

\textbf{Online matching where items arrive over time.}
Motivated by online platforms, many models where items arrive over time have been recently studied, with constant-factor approximations \citep{aouad2022dynamic,kessel2022stationary} and optimal algorithms \citep{aveklouris2021matching,kerimov2021optimality} known under certain regimes.
These papers focus on steady-state behavior, which is possible because items are arriving indefinitely.
Our paper contrasts these models because there is still a finite supply of items; they merely need to "arrive" to acknowledge each customer, and we provide a constant-factor approximation for any finite time horizon and market size.

\section{Problem Definition and Notation}
\label{sec:prelim}

We use $\bpgraph$ to denote a bipartite graph with vertex set $U\cup V$ and edge
set $E\subseteq U\times V$. Let $U=\set{u_{1},\dots,u_{m}}$ represent offline vertices and
$V=\set{v_{1},\dots,v_{n}}$ represent online vertices. For an edge $e=(u, v)$, we
denote the \emph{weight} of edge $e$ by $w_{e}$ or $w_{u,v}$; for the special
case of vertex weights, each offline vertex $u_{i}$ has a weight denoted by
$w_{i}$, and $w_{u_{i},v} = w_{i}$ for all $v\in V$.
We will generally consider \textit{stochastic edges}, which means that
for each edge $(u_{i},v_{j})\in U\times V$, there is a known probability $p_{i,j}$ with which that edge with independently exist when probed.

When considering the online matching problem for a single online vertex (customer) $v$, we will refer to it as a \textit{star graph}.  In this case, we simplify notation and write $p_{i}$ to denote
the probability of edge $(u_{i},v)$. We also use $p_{u,v}$ for the given
probability of edge $(u, v)$ when indices $i$ and $j$ are not required. Without loss of generality, we may assume that $p_{u,v}$ is defined even for $(u,v)\notin E$, since in this case we can simply let $p_{u,v} = 0$.

We are further given a \emph{patience} value $\patience_{v}$ for each online vertex
in $V$ (we may also write $\patience_{j}$ for the patience of vertex
$v_{j}\in V$) that signifies the number of times we are allowed to probe
different edges incident on $v$ when it arrives. Each edge may be probed at most
once and if it exists, we must match it and stop probing (probe-commit model).

We consider the online vertices arriving at positive integer \emph{times}. In the
adversarial arrival model, the vertices of $V = \set{v_{1},v_{2},\dots,v_{n}}$ are
fixed and the order of their arrival is set by an adversary so as to minimize
the expected matching weight. We assume (wlog) that the vertices arrive in the
order $v_{1},v_{2},\dots,v_{n}$. When we consider the stochastic arrival models, $V$ instead specifies a set of vertex \emph{types}, and at time
$t$, a vertex of an independently randomly chosen type from a known distribution arrives.
Generally these distributions can vary across time, which we call the \textit{prophet arrival}\footnote{This is because it was the arrival model of original focus in prophet inequality papers \citep{krengel1977semiamarts}.} model; we also consider the special case where these distributions are identical, which we refer to as \textit{IID arrivals}.
In these models, we will let $T$ denote the length of the time horizon which is assumed to be known (otherwise the problem is impossible; see \Cref{sec:unknown_pat}).

When an online vertex $v$ arrives at time $t$, we attempt to match it to an
available offline vertex. We are allowed to probe edges incident to $v_{t}$
one-by-one, stopping as soon as an edge $(u_{i},v_{t})$ is found to exist, at
which point the edge is included in the matching and we receive a reward of
$w_{i}$. We are allowed to probe a maximum of $\patience_{t}$ edges (in the
stochastic patience models, $\patience_{t}$ is not known a priori and is
discovered only after $\patience_{t}$ failed probes); if $\patience_{t}$ edges
are probed and none of the edges exist, then vertex $v_{t}$ remains unmatched
and we receive no reward. If we successfully match $v_{t}$ to $u_{i}$, we say
that $w_{i}$ is the \emph{value} or \emph{reward} of $v_{t}$'s match; if $v_{t}$ remains
unmatched, we say it has a value or reward of $0$. The next online vertex
$v_{t+1}$ does not arrive until we have finished attempting to match $v_{t}$
(either by exhausting the patience constraint, or by successfully matching
$v_{t}$). Thus, there is only ever one online vertex available for matching at
any one time.

We use $G$ to denote an instance, which includes the graph, weights, edge probabilities, and any arrival distributions including patience.
Given any instance $G$ one can consider an optimal \textit{offline} algorithm, which knows in advance the online vertices that will arrive.
For the adversarial arrival model, this means knowing the sequence  $v_{1},v_{2},\dots,v_{n}$;
for the stochastic arrival models, this means knowing the sequence of $T$ types that will be realized.
We let $\OPT(G)$ denote the expected reward collected by the best sequential probing algorithm on $G$ that has access to this offline information, noting that: i) $\OPT(G)$ does not know the realizations of the stochastic edges in advance either; ii) computing $\OPT(G)$ is difficult but unnecessary; iii) for stochastic arrival models, this expectation is also over the realizations of the $T$ types; and iv) we assume that the offline algorithm must finish\footnote{We note that \citet{borodin2021prophet} derive results against a stronger benchmark, which can switch back-and-forth between online vertices.} the interactions with one online vertex before moving to the next.
Meanwhile, we let $\ALG(G)$ denote the expected reward collected by a fixed online algorithm on $G$, again taking a realization over types in the stochastic arrival models, and any potential randomness in the algorithm.

With this understanding, we say that a fixed (potentially randomized) online algorithm is \textit{$c$-competitive} if $\ALG(G)/\OPT(G)\ge c$ for all instances $G$, where $c$ is a constant in [0,1].  We are interested in the maximum value of $c$ for which an algorithm can be $c$-competitive, which is referred to as the \textit{competitive ratio}.

\subsection{Outline for the Rest of the Paper}
Our main algorithms and results for online matching with stochastic edges are presented in
\Cref{sec:lpalgs}. In that section, we first present an algorithm for the
vertex-weighted case, under adversarial arrivals, and show that it is
$1/2$-competitive. To our knowledge, this is the first result for this setting.
In addition, we provide an algorithm for the \emph{edge-weighted} case, under
prophet arrivals. Here too, we are able to show the algorithm is
$1/2$-competitive; we further show that a slight modification can improve the
competitive ratio to $1-1/\euler$ when either the edge-weighted assumption is
relaxed to vertex weights or the non-stationary assumption is relaxed to known
IID arrivals.

% Before we can fully describe the algorithms and results for online matching,
% however, we must solve a simpler problem which can then be used as a subroutine
% in our algorithms. Thus, in \Cref{sec:stargraphs}, we consider the problem of a
% star graph, which corresponds to the case of a single online customer. For this
% problem, when the patience of the customer is known, there is an optimal
% algorithm due to~\citet{purohit2019hiring}.

All of the algorithms of \Cref{sec:lpalgs} rely on utilizing, as a black box, an algorithm for
the simpler problem of a star graph, which corresponds to a single online
customer. For this problem, when the patience of the customer is known, there is
an optimal algorithm based on dynamic programming due
to~\citet{purohit2019hiring}. However, the results in \Cref{sec:lpalgs} are
stated in an abstract general manner, which allows us to swap out the algorithm
of~\citet{purohit2019hiring} for algorithms solving the star graph problem under
different settings of patience. In \Cref{sec:stargraphs}, we introduce new,
stochastic, models for the patience of the customer, and give algorithms for
these new settings. These new star graph algorithms can then be used as black
boxes in the algorithms of \Cref{sec:lpalgs}, giving results for online matching under these new patience models.

Finally, in \Cref{sec:neg_results}, we present our negative results following the order described earlier.

\section{Algorithms for Online Matching with Stochastic Edges}
\label{sec:lpalgs}%
In this section, we present our results for online matching with stochastic
  edges. Recall that in this problem, the items $U$ are known in advance while
  the customers arrive one by one in an online fashion. When a customer of type $v$
  arrives, we learn the probability $p_{u,v}$, for each $u\in U$, that the customer
  will purchase item $u$ if offered; if purchased, we gain some reward specified
  by $w_{u,v}$, and if not, we may proceed to offer another item up to a total of
  $\pat_{v}$ offers. In the simplest setting, $\pat_{v}$ is known to the
  algorithm, although our framework can also handle  settings where only a
  probability distribution over $\pat_{v}$ is known (see \Cref{sec:stargraphs}). An online algorithm must
  make offer items sequentially to the customer, and must make all its offers
  before the next one arrives. The goal is to maximize the expected total reward
  across all customers.

\subsection{Vertex-Weighted, Adversarial Arrivals} \label{sec:adversarial}%
We present a greedy algorithm \texttt{AdvGreedy} which achieves a $0.5$-approximation for
online matching with vertex weights, stochastic rewards, and patience
constraints in the adversarial arrival model. Recall that in this setting,
  each offline item $u_{i} \in U$ has a weight $w_{i}$ such that
  $w_{u_{i},v} = w_{i}$ for all customers $v\in V$, modeling the situation where
  items have a fixed price that is the same for all customers. If vertex $v$ is
the $t^\text{th}$ vertex to arrive online, we say $v$ \emph{arrives at time $t$}.
Recall that the goal of the problem is to offer items to customer $v_{t}$ one by
one, until either the patience runs out or an item is successfully sold; and the
entirety of this process is carried before the $(t+1)^{\text{st}}$ arrival,
$v_{t+1}$, arrives. The goal is to maximize the total revenue across all
arrivals. Our algorithm makes use of a black box subroutine, \textsc{StarBB}, which takes
as input a single online vertex, along with the probabilities and weights of its
incident edges; $\textsc{StarBB}(v, \mathbf{p}, \mathbf{w})$ simply probes edges incident
to online vertex $v$ in some order, until either $v$'s patience is exhausted or
a match is successful. Our results hold as long as \textsc{StarBB} is optimal, or within
a constant factor $\kappa$ of optimal. When the patience is deterministic and known
to the algorithm, we can use the dynamic programming-based algorithm of~\citet{purohit2019hiring} for
\textsc{StarBB}; the dynamic program gives a sequence of $\pat_{v}$ vertices in $U$ to
probe in this order. This is optimal for star graphs, so $\kappa=1$. In \Cref{sec:stargraphs}, we
extend this online stochastic matching result to settings where the patience is
unknown and stochastic, by giving star graph algorithms for these settings and
proving constant factor approximations for them.

\begin{algorithm}
\caption{Use a Star Graph Black Box to greedily match arriving vertices}
\label{alg:dpgreedy}
\DontPrintSemicolon
\SetKwFunction{AdvG}{AdvGreedy}
\SetKwProg{Fn}{Function}{:}{}
\Fn{\AdvG{$U$, $V$, $\mathbf{p}$, $\mathbf{w}$}} {
\For{Arriving vertex $v\in V$
}{
$%(i^{*}_{1},\dots,i^{*}_{m})\gets
\Call{StarBB}{v,
\mathbf{p}, \mathbf{w}}$ \;
% \For{$\pat:=1$ to $m$} {
% Probe edge $(u_{i^{*}_{\pat}},v)$ \;
% }
}
}
\end{algorithm}

Let $\ALG(G)$ denote the expected size of the matching produced by this
algorithm on the graph $G$. Let $\OPT(G)$ denote the expected size of the
matching produced by an optimal \emph{offline} algorithm. Our main result here is:
\begin{theorem} \label{thm:halfcr}
Given a $\kappa$-approximate black box for solving star graphs, \Cref{alg:dpgreedy} achieves a competitive ratio of $0.5\kappa$; that is, for any bipartite graph $G$, $\frac{\ALG(G)}{\OPT(G)}\geq 0.5\kappa$.
\end{theorem}
To prove this, we first present an LP which provides an upper bound on the offline optimal.

%%%---NEW LP BOUND---%%%

\subsubsection{An LP upper bound on $\OPT(G)$.} \label{sec:advlp}
We formulate a new LP for our problem by adding a new constraint to, \eqref{eq:lpadv-dpopt}, to the standard LP relaxation of the problem. This new LP gives a tighter upper bound on the offline optimal solution, $\OPT(G)$. Note that our algorithm does not need to solve this new LP, as it is only used in our analysis.

\begin{subequations} \label{lp:adv}
\begin{alignat}{3}
\LPOPT := \max &\sum_{u\in U} \sum_{v\in V}
x_{u,v}p_{u,v} w_{u} \tag{\ref{lp:adv}} \\[2ex]
\text{subject to }
& \sum_{v\in V} x_{u,v}p_{u,v} \le 1 & &\hspace{2em}\forall u\in U \label{eq:lpadv-1matchu} \\
& \sum_{u\in U} x_{u,v} p_{u,v} \le 1 & &\hspace{2em}\forall v\in V \label{eq:lpadv-1matchv} \\
& \sum_{u\in U} x_{u,v} \le \bE[\pat_{v}] & &\hspace{2em}\forall v\in V \label{eq:lpadv-tprobes} \\
&\sum_{u\in U'} x_{u,v}p_{u,v}w_{u} \le \OPT(U',v) & &\hspace{2em} \forall U'\subseteq U, v\in V \label{eq:lpadv-dpopt} \\
& 0 \le x_{u,v} \le 1 & &\hspace{2em}\forall u\in U, v\in V \notag
\end{alignat}
\end{subequations}
In this LP, we slightly abuse notation and write $\OPT(U',v)$ to denote $\OPT(G')$ for a star graph $G' = (U', \set{v}, U'\times\set{v})$. Recall that $OPT(U',v)$ can be computed by a black box.

We first show in \Cref{sec:deferred_proofs} that our strengthened LP is a valid upper bound.
\begin{lemma} \label{lem:lpopt}
For any bipartite graph $G$, $\LPOPT(G) \geq \OPT(G)$.
\end{lemma}

\subsubsection{Proof of $0.5$-competitiveness.}
We can then bound the performance of our greedy algorithm relative to the
solution of Linear Program~\eqref{lp:adv}.
%\Cref{lem:lpopt} then implies that this bounds the competitive ratio. In particular, the following lemma, along with \Cref{lem:lpopt}, implies \Cref{thm:halfcr}.

\begin{lemma} \label{lem:halflp} If $\textsc{StarBB}$ is a $\kappa$-approximate
algorithm for star graphs, then for any bipartite graph $G$,
$\ALG(G) \ge 0.5\kappa\LPOPT(G)$.
\end{lemma}
By \Cref{lem:lpopt,lem:halflp}, we have $ \ALG(G) \ge 0.5\kappa\LPOPT(G) \ge 0.5\kappa\OPT(G), $ which implies our
main result of a $\frac{\kappa}{2}$-competitive algorithm. As previously mentioned,
using the probing order given by the dynamic program of \citet{purohit2019hiring} as \textsc{StarBB} gives
$\kappa=1$, so we have a $\frac{1}{2}$-approximation. In \Cref{sec:stargraphs}, we present star graphs
for stochastic patience settings, where $\kappa$ is not necessarily $1$.
%\nnote{TODO: Remove the following sentence, and instead add discussion of these bounds in the star graph section.}
%Note that $\kappa=1$ if patience is known or follows the constant hazard rate model, and $\kappa=1/2$ for an arbitrary stochastic patience.  Thus, we can achieve a competitive ratio of 1/2 in the former case and a competitive ratio of 1/4 in the latter case.
%The same argument can be used to show that for unknown patience, a $\kappa$-approximate solution to the star graph problem yields a $\kappa/2$-competitive matching in the online setting with adversarial arrivals. Thus, we can achieve a $1/2$ competitive ratio for constant hazard rates and $1/4$ for arbitrary patience distributions.

\subsection{Prophet Arrival Setting} \label{sec:edgeprophet}

If we wish to allow for arbitrary edge weights, we must consider a different
arrival model. A popular arrival model in the literature is the known IID
setting, as discussed in~\Cref{sec:relatedWork}; here, we consider a
generalization of the known IID setting, which we call the \emph{prophet arrival
model}. In this model, $V$ specifies a set of possible arrival \emph{types};
each arrival takes on one of these types randomly, according to a known
distribution. The probabilities at each arrival are independent of previous
arrivals, and the distribution over possible types can be different at each
arrival.

For $t=1,2,\dots,T$ and $v\in V$, denote by $q_{tv}$ the probability that the vertex arriving at time $t$ will be of type $v$. For convenience, we denote by $q_{v} = \sum_{t=1}^{T} q_{tv}$ the expected number of arrivals of a vertex of type $v$.

We employ a new exponential-sized LP relaxation.
In this LP, the variables correspond to \emph{policies} for probing an arriving online vertex.
A deterministic policy $\pi$
for matching any online vertex type $v$ is characterized by a permutation of some \emph{subset of $U$}.
The policy specifies the strategy of attempting to match $v$ to vertices of $U$ in the order given by $\pi$, until either a probe is successful, all vertices in $\pi$ are attempted, or the patience of $v$ is exhausted. Let $\cP$ denote the set of all deterministic policies.

We present our LP in~\eqref{lp:prophet} below.  We let $p_{uv}(\pi)$ denote the probability that an online arrival of type $v$ is matched to offline vertex $u$ when following policy $\pi$, assuming that all vertices of $\pi$ are still unmatched. The decision variables in the LP are given by $x_v(\pi)$, and we let $\OPTP(G)$ denote the true optimal objective value of the exponential-sized LP.
\begin{subequations} \label{lp:prophet}
\begin{alignat}{3}
\OPTP=\max
& \sum_{v\in V}\sum_{\pi\in \cP}x_{v}(\pi) \sum_{u\in U} p_{uv}(\pi) w_{uv} \tag{\ref{lp:prophet}} \\[2ex]
\text{subject to }
& \sum_{v\in V}\sum_{\pi\in\cP}x_{v}(\pi)p_{uv}(\pi) \le 1& &\hspace{0em}\forall u\in U \label{eq:prophetlp-offline} \\
& \sum_{\pi\in\cP}x_{v}(\pi) = q_{v} & &\hspace{0em}\forall v\in V \label{eq:prophetlp-polnum} \\
& x_{v}(\pi) \ge 0 & &\hspace{0em}\forall v\in V, \pi\in\cP
\end{alignat}
\end{subequations}

We can interpret $x_{v}(\pi)$ as the expected number of times policy $\pi$ will be applied to an online vertex of type $v$.
Constraint~\eqref{eq:prophetlp-offline} says that in expectation each offline vertex $u$ can be matched at most once.
Constraint~\eqref{eq:prophetlp-polnum} comes from the fact that exactly one policy (possibly the policy that makes zero probes) must be applied to each arriving vertex of type $v$.
It follows from standard techniques \citep[see e.g.][Lemma~9]{BansalLPCure} that even for a clairvoyant, who knows the realized types of arrivals in advance, if we let $x_v(\pi)$ denote the expected number of times it applies policy $\pi$ on an online vertex of type $v$, then this forms a feasible solution to the LP with objective value equal to the clairvoyant's expected weight matched.
Therefore, if we can bound the algorithm's expected reward relative to $\OPTP(G)$, then this would yield a competitive ratio guarantee.

\subsubsection{Solving the exponential-sized LP.} \label{sec:solvingExpLP}
The Linear Program~\eqref{lp:prophet} has a variable $x_{v}(\pi)$ for every possible online vertex type $v$ and policy $\pi\in\cP$. If $\cP$ has polynomial size, this LP therefore has polynomial size. This happens, for example, when all online vertices have a small constant patience. For instance, when the patience of all vertices is $1$, a policy $\pi\in\cP$ is characterized by a single offline vertex and so $|\cP| = |U|$.

However, for general patience values the LP has exponentially many variables.
Nonetheless, since there are only a polynomial number of constraints, a sparse solution which is polynomially-sized still exists.
To help in solving the LP, we consider the dual.
\begin{subequations} \label{lp:prophetdual}
\begin{alignat}{3}
\text{minimize }
& \sum_{u\in U}\alpha_{u} + \sum_{v\in V} q_{v} \beta_{v} \tag{\ref{lp:prophetdual}} \\[4ex]
\text{s.t. }
& \sum_{u\in U} p_{uv}(\pi)\alpha_{u} + \beta_{v} \ge \sum_{u\in U} p_{uv}(\pi) w_{uv} & &\hspace{6em}\forall v\in V, \pi\in\cP \label{eq:prophetdual-main}\\
& \alpha_{u} \ge 0 & &\hspace{6em}\forall u\in U
\end{alignat}
\end{subequations}

Note that the exponential family of constraints~\eqref{eq:prophetdual-main} can be re-written as
\begin{align} \label{eq:prophetdualmax}
\max_{\pi\in\cP}\sum_{u\in U} p_{uv}(\pi)(w_{uv} - \alpha_{u}) \le \beta_v, && \forall v\in V
\end{align}
which is equivalent to, for every online type $v\in V$, solving the probing problem for a star graph consisting of a single online vertex of type $v$ and adjusted edge weights $w'_{uv}=w_{uv}-\alpha_u$.

%%%% INPUT
Recall that if for every online type $v\in V$, we have a black box which can solve the maximization problem in~\eqref{eq:prophetdualmax} to optimality, then we have a separation oracle for the dual LP which can either establish dual feasibility or find a violating constraint given by $v$ and $\pi$.  By the equivalence of separation and optimization, this allows us to solve both the dual LP~\eqref{lp:prophet} and in turn the primal LP~\eqref{lp:prophetdual} using the ellipsoid method, as formalized in the proposition below.
\begin{proposition} \label{prop:ellipsoid}
Given black box star graph algorithms for every online vertex type which can verify~\eqref{eq:prophetdualmax} in polynomial time, the exponential-sized LP~\eqref{lp:prophet} can be solved in polynomial time.
\end{proposition}

An explicit statement which directly establishes \Cref{prop:ellipsoid} can be found in Chapter~14 of %Schrijver
\citet{schrijver1998theory}.
We should note that the ellipsoid method is only used to establish poly-time solvability in theory, and that a column generation method which adds primal variables $x_v(\pi)$ as needed on-the-fly is usually more efficient in practice.

%In either case, we only have a black box for verifying~\eqref{eq:prophetdualmax} if every online type $v\in V$ has a patience which is deterministic \citep{purohit2019hiring} or stochastic with constant hazard rate (\Cref{sec:starconsthazard}), or has a small number of neighbors.
In either case, while the dynamic programming approach of
\citet{purohit2019hiring} gives us an optimal algorithm for
verifying~\eqref{eq:prophetdualmax} in polynomial time, for the setting of
general stochastic patience described in \Cref{sec:arbpatience}, the algorithm
we present is only a $1/2$-approximation for the star graph problem.
Nonetheless, one can still use the structure of the dual LP, along with this
approximate separation, to compute a $(1/2-\epsilon)$-approximate solution to the
primal LP.\@ This is formalized in the \namecref{prop:potentialBased} below.

\begin{proposition} \label{prop:potentialBased}
Suppose for every type $v\in V$, we are given a black box which is a $\kappa$-approximation algorithm for the maximization problem in~\eqref{eq:prophetdualmax}.
Then a solution to LP~\eqref{lp:prophet} with objective value at least $(\kappa-\epsilon)\OPTP(G)$ can be computed, in time polynomial in the problem parameters and $1/\epsilon$.
\end{proposition}

\Cref{prop:potentialBased} follows from the \emph{potential-based framework} of \cite{CSL16}, which is
elaborated on in \Cref{sec:potential_based}. Armed with
\Cref{prop:potentialBased,prop:ellipsoid}, we can use any optimal or
approximately optimal algorithm for star graphs to obtain an (approximately)
optimal solution to LP~\eqref{lp:prophet}, which we can then use to solve the
problem of online matching under prophet arrivals.

\subsubsection{Algorithm and analysis based on exponential-sized LP.}
We now show how to use the LP~\eqref{lp:prophet} to design a $\kappa/2$-competitive online
algorithm, given a feasible LP solution $\ssx_v(\pi)$ which is at least $\kappa\cdot\OPTP(G)$,
for some $\kappa\le1$. For each vertex $u\in U$, let
$\ssw_u=\sum_{v\in V}w_{uv}\sum_{\pi\in\cP}p_{uv}(\pi)\ssx_v(\pi)$ denote the expected reward of
matching $u$ according to the assignment $\ssx$. Notice that the objective value
of the given solution is $\sum_{u\in U} \ssw_{u}$, which is at least $\kappa\cdot\OPTP(G)$. Using
this notation, our algorithm is given in \Cref{alg:prophet}. By LP constraint~\eqref{eq:prophetlp-polnum}, we
have $\sum_{\pi\in\cP} \ssx_v(\pi)/q_v = 1$, so the probability distribution over
policies defined in \Cref{alg:prophet} is proper. We also note that since at most
polynomially many variables $\ssx_v(\pi)$ will be nonzero, this distribution has
polynomial-sized support and can be sampled from in polynomial time. The
algorithm itself is fairly simple, selecting a policy $\pi$ at random with
probability proportional to LP solution $x^{*}_{v}(\pi)$ upon the arrival of a
vertex $v_{t}$ of type $v$. Then, policy $\pi$ is followed in order to attempt to
match the online vertex, but with two quirks: first, we skip probing any edge
$(\pi_{i},v_{t})$ for which the weight of the edge is too small (specifically, if
$w_{\pi_{i},v} < \ssw_{\pi_{i}}/2$), and second, if $\pi$ tells us to probe an edge
$(\pi_{i},v_{t})$ for which offline vertex $\pi_{i}$ is unavailable, we "simulate"
the probing instead, terminating with no reward if the simulated probe is
successful. This simulation technique was used in \citet{brubach2017attenuate}
and is also used by our star graph algorithm in \Cref{sec:arbpatience}.

\begin{algorithm}
\caption{Random Arrivals from Known Distributions (Prophet Arrivals)}
\label{alg:prophet}
\DontPrintSemicolon
\SetKwFunction{OM}{OnlineMatch}
\SetKwProg{Fn}{Function}{:}{}
\Fn{\OM{$U$, $V$, $\mathbf{p}$, $\mathbf{w}$}}{
\For{$t=1$ to $T$}{
Online vertex $v_t$, of type $v\in V$, arrives \;
Choose a policy $\pi = (\pi_1, \pi_2, \dots, \pi_{\ell})$ with probability $\ssx_v(\pi)/q_v$ \;
\For{$i:=1$ to $\ell$} {
\If{$w_{\pi_i,v} < \ssw_{\pi_i}/2$}{
Skip to next $i$ \;
}
\ElseIf{$\pi_i$ is unmatched}{
Probe edge $(\pi_i,v_t)$ and match if successful for reward $w_{\pi_i,v}$ \;
}
\Else{
Simulate probing $(\pi_i,v_t)$. If successful, move to next arrival without matching $v_t$. \;
}
}
}
}
\end{algorithm}

\begin{theorem} \label{thm:prophethalf}
Under general edge weights and known arrival distributions, \Cref{alg:prophet} is $\kappa/2$-competitive for the online bipartite matching with patience problem,
assuming we are given a solution to LP~\eqref{lp:prophet} with objective value at least $\kappa\cdot\OPTP(G)$.
\end{theorem}

The proof of \Cref{thm:prophethalf} is deferred to \Cref{sec:deferred_proofs}. In the case of deterministic,
known patience, the algorithm of~\citet{purohit2019hiring} can be used as the
black box star graph algorithms to allow verifying~\eqref{eq:prophetdualmax},
and thus solving LP~\eqref{lp:prophet} exactly. In this case, $\kappa=1$ in
\Cref{thm:prophethalf}, and we have a $1/2$-competitive algorithm for online matching with patience, under prophet arrivals and arbitrary edge
weights. In \Cref{sec:stargraphs}, we extend the classical online stochastic
matching problem to consider stochastic patience, and give star graph algorithms
with constant factor approximation guarantees, allowing us to apply
\Cref{thm:prophethalf} to these settings as well.

\subsection{Improvement for IID Arrivals} \label{sec:prophetiid}

In the case of IID arrivals, i.e., where $q_{1v} = q_{2v} = \dots = q_{Tv}$ for all vertex types $v\in V$, a slight modification to \Cref{alg:prophet} yields a competitive ratio of $(1-1/e)\kappa$ when given a feasible solution to LP~\eqref{lp:prophet} that is at least $\kappa\cdot \OPTP(G)$ of optimal.

The only change to the algorithm from the previous non-IID setting is that for IID arrivals, we do not skip probing vertices $u\in U$ when $w_{uv} < \ssw_u / 2$.
The full pseudocode is given in \Cref{alg:prophetiidvw}.

\begin{algorithm}
\caption{Random Arrivals from Known Identical Distributions (IID Arrivals)}
\label{alg:prophetiidvw}
\DontPrintSemicolon
\SetKwFunction{OM}{OnlineMatch}
\SetKwProg{Fn}{Function}{:}{}
\Fn{\OM{$U$, $V$, $\mathbf{p}$, $\mathbf{w}$}}{
\For{$t=1$ to $T$}{
Online vertex $v_t$, of type $v\in V$, arrives \;
Choose a policy $\pi = (\pi_1, \pi_2, \dots, \pi_{\ell})$ with probability $\ssx_v(\pi)/q_v$ \;
\For{$i:=1$ to $\ell$}{
\If{$\pi_i$ is unmatched}{
Probe edge $(\pi_i,v_t)$ and match if successful for reward $w_{\pi_i,v_t}$ \;
}
\Else{
Simulate probing $(\pi_i,v_t)$. If successful, move to next arrival without matching $v_t$. \;
}
}
}
}
\end{algorithm}

\begin{theorem} \label{thm:prophetiid}
Under general edge weights and known IID arrivals, \Cref{alg:prophetiidvw} is $\kappa(1-1/\euler)$-competitive for the online bipartite matching with patience problem,
assuming we are given a solution to LP~\eqref{lp:prophet} with objective value at least $\kappa\cdot\OPTP(G)$.
\end{theorem}

The proof of \Cref{thm:prophetiid} is deferred to \Cref{sec:deferred_proofs}. As with \Cref{thm:prophethalf},
the dynamic program of \citet{purohit2019hiring} can be used to get an optimal
LP solution, so that \Cref{thm:prophetiid} gives a competitiveness of
$1-1/\euler$ for \Cref{alg:prophetiidvw} under the classical deterministic
patience setting. \Cref{sec:stargraphs} extends the result to online stochastic
matching with stochastic patience, by providing black box algorithms for those
settings which can then be used to find approximately optimal solutions to
LP~\eqref{lp:prophet} as per \Cref{prop:potentialBased}.

\subsection{Improvement for Vertex Weights}\label{sec:prophetvw}
With a new analysis, we can show that
\Cref{alg:prophetiidvw} still achieves a competitive ratio of $1-1/\euler$ in the prophet (non-stationary) setting in the case of \emph{vertex weights}.

\begin{theorem} \label{thm:prophetvw}
Under vertex-weights and known arrival distributions, \Cref{alg:prophetiidvw} is $\kappa(1-1/e)$-competitive for the online bipartite matching with patience problem,
assuming we are given a solution to LP~\eqref{lp:prophet} with objective value at least $\kappa\cdot\OPTP(G)$.
\end{theorem}

The proof of \Cref{thm:prophetvw} is deferred to \Cref{sec:deferred_proofs}. It has identical
implications as \Cref{thm:prophetiid}, with the improvement beyond the 1/2-guarantee of \Cref{thm:prophethalf} coming
from the vertex-weighted assumption instead of the IID assumption.

\section{Algorithms for Star Graphs} \label{sec:stargraphs} %
All of the algorithms in \Cref{sec:lpalgs} make use of a \emph{black box algorithm} for solving the
special case of a single customer (called the "star graph" problem, since it
corresponds to a star graph with the customer as the center vertex): Under
adversarial arrivals, the black box is used for each arriving vertex, while for
prophet arrivals it is used as a separation oracle for the dual LP.\@ For the
classic problem of online matching with stochastic edges and finite patience, we
may use the dynamic program of~\citet{purohit2019hiring} as this black box. Since this dynamic program
is optimal for star graphs, $\kappa=1$ in all of the results of \Cref{sec:lpalgs}. In this
\namecref{sec:stargraphs}, we consider the problem where the patience of the
customer is not known, but is instead determined by some stochastic process.
These algorithms can be used as black boxes for the algorithms in \Cref{sec:lpalgs}, giving
competitive ratios for online matching with stochastic edges (with multiple
customers) under these new stochastic patience settings. Finally, we also give
an algorithm for an alternate setting where the customer has a deterministic
patience, but items arrive over time according to Bernoulli processes.
Throughout, we use $m$ to denote the number of items, as in \cref{sec:prelim}.

%Note that the edge-weighted and vertex-weighted problems are equivalent on star graphs.

%\subsection{Known Patience} \label{sec:starknown}
%
%
%An algorithm to solve the known patience case via dynamic program was given in~\citet{purohit2019hiring} (see section~2.1 of that paper, i.e.\ the $k=1$ case of the hiring problem they study). We present this briefly here but do not discuss it in detail.
%
%Recall that we assume wlog that $w_{1} \ge w_{2} \ge \dots \ge w_{m}$. To simplify notation, we write $p_i$ for $p_{u_i, v}$. We then define $f(i, \pat)$ to be the maximum possible expected value if we may only match $v$ to one of the neighbors in $\set{ v_{i}, \dots, v_{m}}$ (i.e., neighbors $i$ through $m$).
%\begin{equation} \label{eq:advdp}
%\begin{split}
%f(i,\pat) =
%\max \{&p_{i} w_i + (1 - p_{i})f(i+1, \pat-1), \\
%&f(i+1, \pat)\}
%\end{split} 
%\end{equation}
%with $f(i,0) = 0$.
%
%
%
%This dynamic program immediately implies the following fact about the known patience setting:
%\begin{theorem}
%\label{thm:optstar}
%There exists an algorithm for the  matching with known patience problem on vertex-weighted star graphs that finds the optimal probing strategy in polynomial time.
%\end{theorem}

\subsection{Constant Hazard Rate} \label{sec:starconsthazard}

In this setting, the patience is
random and unknown, with a constant "hazard rate" $r_{i}$ for each $i\in[m]$.\footnote{Throughout this paper, we use the notation $[m] := \{1, 2, \dots, m\}$}
When an attempted match with $u_{i}$ is unsuccessful, the customer $v$
runs out of patience with probability $r_{i}$ and remains available for
another match attempt with probability $1-r_{i}$.
% The "hazard rate" $1-r$ is the probability that $v$ becomes unavailable after a failed probe. Equivalently, the patience is drawn from a negative binomial distribution, $\pat_v\sim\mathrm{NB}(1,1-r)$.
When the hazard rate $r_{i}$ is the same for all $i\in[m]$ (i.e., $r_{i}=r$ for some $r$ and every $i$), this is equivalent to the patience having a constant hazard rate of $r$.

\begin{theorem} \label{thm:constantHazard}%
For maximizing the expected weight of the matched item in the (item-dependent)
Constant Hazard Rate patience model, it is optimal to probe items in decreasing order of
\[
\frac{w_{i}p_{i}}{p_{i} + (1-p_{i})r_{i}}.
\]
\end{theorem}
\Cref{thm:constantHazard} tells us that, in the special case where the patience
distribution can be modeled with individual per-item hazard rates, we can
achieve an optimal star graph probing strategy. As such, we can apply all of the algorithms and results from \Cref{sec:lpalgs} with $\kappa=1$.

\subsection{Arbitrary Patience Distributions} \label{sec:arbpatience}

We address the case where the patience is stochastic (and unknown) and can
follow an arbitrary known distribution. Without loss of generality, we can
assume the patience distribution has finite support (more specifically, that the
patience $\pat \in
[m]$); %\footnote{We use the notation $[m]$ to denote the set $\{1, 2, \dots, m\}$.}
%\textcolor{red}{[ NPG: Removed a footnote about the notation $[m]$ since it's
%  used before this point. Should we define it earlier, or not bother since it's
%  somewhat standard anyway? ]} \wnote{I would define it at the first occurrence.}); 
this is because a patience greater than $m$ is
equivalent to a patience of $m$ (additionally, a patience of $0$ indicates a
vertex which can never be matched by any algorithm, since no probes can be made,
so we can simply ignore such vertices). We use an LP-based approach for this
problem. We denote by $q_{\pat}$ the probability that the patience of the online
vertex $v$ is \emph{at least} $\pat$. Notice that $1=q_{1} \ge q_{2} \ge \ldots \ge q_{n}$. Our
approach utilizes the LP~\eqref{lp:arbpatience} described in the next paragraph. The variables are
$x_{j\pat}$, which correspond to the probability of \emph{attempting to match} with $j$
on the $\pat^{\text{th}}$ attempt. The value $s_{\pat}$ represents the
probability that the online vertex is available for a $\pat^{\text{th}}$ match
attempt, meaning its patience is at least $\pat$ and all previous match attempts
were unsuccessful. This can be calculated from the $x_{j\pat}$, $q_{\pat}$, and
$p_{j}$ values.
\begin{subequations} \label{lp:arbpatience}
\begin{alignat}{3}
\max
& \sum_{j=1}^{m}w_{j}p_{j}\sum_{\pat=1}^{m}x_{j\pat} \tag{\ref{lp:arbpatience}} \\[2ex]
\text{subject to }
& \sum_{\pat'=\pat}^{m}x_{j\pat'} \le s_{\pat},& &\hspace{2em}\forall j\in\set{1,2,\dots,m}, \pat\in\set{1,2,\dots,m} \label{eq:arblp-1pernode} \\
& \sum_{j=1}^{m}x_{j\pat} \le s_{\pat},& &\hspace{2em}\forall \pat\in\set{1,2,\dots,m} \label{eq:arblp-avail} \\
& x_{j\pat}\ge 0,& &\hspace{2em}\forall j\in\set{1,2,\dots,m},\pat\in\set{1,2,\dots,m} \notag \\[1ex]
& s_{1} = 1, \label{eq:sUpdateOne} \\
& s_{\pat} = \frac{q_{\pat}}{q_{\pat-1}} \left( s_{\pat-1} -
\sum_{j=1}^{m} p_{j}x_{j,\pat-1} \right) & & \hspace{2em}\forall \pat\in\{2,\ldots,m\} \label{eq:sUpdateTwo}
\end{alignat}
\end{subequations}

% \nnote{I tried to rewrite the explanations for the constraints, since at least
%   one reviewer said they were confused by the LP.\@ I'm not sure how much of an
%   improvement this is over what we had originally. Any further edits or
%   improvements to this are welcome. (The old description is commented out in the
%   tex source, for reference)}%
As an example, consider the e-commerce application, where we wish to offer items
to a customer one by one. The variable $x_{j\pat}$ indicates the probability
that the customer is offered item $j$ after $\pat-1$ other items have already
been offered (and rejected). This can only happen if the customer's patience is
at least $\pat$ \emph{and} the customer has not already purchased an item (prior to the
$\pat^{\text{th}}$ offer), since otherwise no offers can be made at this point.
The value $s_{\pat}$ denotes precisely this probability, that is the probability
that the customer's patience is at least $\pat$ and the customer has rejected
all previous offers. This suggests a simple constraint: $x_{j\pat} \le s_{\pat}$.
However, two stronger conditions can be given: first, no valid strategy can
offer item $j$ at any point on or after the $\pat^{\text{th}}$ attempt, if the
customer has patience less than $\pat$ or purchases an item offered before the
$\pat^{\text{th}}$. Summing over all offers after $\pat-1$,
$\sum_{\pat'=\pat} x_{j\pat'}$ is the probability that item $j$ is offered at or
after the $\pat^{\text{th}}$ step, and the reasoning above gives us the
constraints~\eqref{eq:arblp-1pernode}. Further, if the customer is unavailable (due to having
purchased an item or running out of patience) for the $\pat^{\text{th}}$
attempt, then no item can be offered on the $\pat^{\text{th}}$ attempt; thus,
the probability of offering an item on attempt number $\pat$ cannot exceed the
probability that the customer is available for the $\pat^{\text{th}}$ attempt.
This gives constraints~\eqref{eq:arblp-avail}. We note that the family of constraints~\eqref{eq:arblp-1pernode} differs
from similar time-indexed LP's in the literature~\citep{ma2014improvements}, in that there is a
constraint for the sum starting at every attempt $\pat$ instead of a single
constraint where $\pat=1$.

Constraints~\eqref{eq:sUpdateOne} and~\eqref{eq:sUpdateTwo} simply give a closed
form expression for the quantities $s_{\pat}$, where
$\frac{q_{\pat}}{q_{\pat-1}}$ is understood to be $0$ if both $q_{\pat}$ and
$q_{\pat-1}$ are $0$. Finally, since $x_{j\pat}$ corresponds to a probability,
we require $x_{j\pat} \in [0,1]$ (we do not explicitly write the constraint
$x_{j\pat} \le 1$ in the LP above, since it is redundant, being implied
by~\eqref{eq:arblp-avail} for $\pat=1$, since $s_{1}=1$). We can see that this
LP upper bounds the optimal algorithm, since taking $x_{j\pat}$ to be the
probability of the algorithm probing $j$ on attempt $\pat$ for all $j$ and
$\pat$, we get a feasible solution to the LP with objective value equal to the
algorithm's expected weight matched.

% Constraint~\eqref{eq:arblp-1pernode} says that each offline vertex $j$
% can be probed at most once across attempts numbered
% $\pat', \pat'+1,\dots,n$ when conditioned on still being available at
% attempt $\pat'$.
% Note that this family of constraints is new in comparison to similar time-indexed LP's which have appeared in the literature \citep{ma2014improvements}, in that there is a constraint for the sum starting at every attempt $\theta'$, instead of just a single constraint where $\theta'=1$.
% Constraint~\eqref{eq:arblp-avail} says that we
% may only attempt the $\pat^{\text{th}}$ match if the online vertex is
% still available for a $\pat^{\text{th}}$ attempt.
% Constraints~\eqref{eq:sUpdateOne}--\eqref{eq:sUpdateTwo} ensure that the values of $s_\pat$ are updated correctly, where the fraction $\frac{q_\pat}{q_{\pat-1}}$ is understood to be 0 if both $q_\pat$ and $q_{\pat-1}$ are 0.
% We can see that this LP upper bounds the optimal algorithm, since taking $x_{j\pat}$ to be the probability of the algorithm probing $j$ on attempt $\pat$ for all $j$ and $\pat$, we get a feasible solution to the LP with objective value equal to the algorithm's expected weight matched.

Our algorithm is simple: we solve LP~\eqref{lp:arbpatience} to
get an optimal solution $x^{*}$, along with the values $s^{*}$. Then,
when making the $\pat^{\text{th}}$ probe, choose each offline vertex $j=1,\ldots,n$
with probability $x^{*}_{j\pat} / s^{*}_{\pat}$ (note that if $s^{*}_{\pat}=0$ then $x^{*}_{j\pat}=0$),
which defines a proper probability distribution by~\eqref{eq:arblp-avail}.
If a vertex $j$ is
chosen to be probed, but has already been unsuccessfully probed in a
previous attempt, we "simulate" probing $j$ instead, and \emph{terminate}
with no reward if the simulated probe is successful.
This simulation technique is important in ensuring that the probability of surviving to the $\theta$'th attempt is consistent with the LP value $\sss_{\pat}$.

\begin{example}
We provide an example to illustrate our LP and algorithm.  Let $m=2$, and suppose there are two items with weights $w_1=1,w_2=2$ and probabilities $p_1=3/4,p_2=1/4$.

First suppose $q_1=q_2=1$, i.e.\ the patience is deterministically 2.  Then the non-zero $x$-values in the optimal LP solution are $x_{2,1}=1,x_{1,2}=0.75$.  This corresponds to probing item 2 (with the higher reward if it succeeds) first, and if the probe fails (occurring w.p.~$3/4$), probing item 1 afterward.  The expected reward is $w_2p_2+(1-p_2)w_1p_1=2/4+3/4\cdot3/4=17/16=1.0625$, which is also the optimal objective value of the LP.

Now consider a more interesting example where $q_1=1,q_2=1/3$.  Then the non-zero $x$-values in the optimal LP solution are $x_{1,1}=0.9,x_{1,2}=0.1,x_{2,1}=0.1$.
% ; this is feasible because $s_2=0.1$.
Our algorithm in this case will probe item 1 first w.p.~0.9, and otherwise (w.p.~0.1) probe item 2 first.
If it survives to the 2nd probe, which occurs w.p.~$s_2=0.1$, it will always probe item 1.
However, note that this will be a "simulated" probe (which generates no reward) unless item 2 was probed on the 1st probe, making item 1 still available for the 2nd probe.
The probability of this occurring is $0.1\cdot3/4\cdot1/3=0.025$.
Therefore, the expected reward of our algorithm is $(0.9+0.025)\cdot3/4+0.1\cdot2/4\approx0.743$.
Meanwhile, the LP optimal value is 0.8.

We note that in this case, the best algorithm is to probe item 1 first (which is the "safe bet", given that the customer has a 2/3 chance of departing after the 1st probe) followed by item 2, which would have expected reward $19/24\approx0.791$, worse than the LP value.
It would have been better than our algorithm though.
However, note that the optimal ordering given many items and an arbitrary patience distribution is generally non-trivial to solve (even in our examples, the optimal ordering switched from 2,1 to 1,2 depending on the patience distribution), and to our knowledge the best approximation algorithm is the 1/2-approximation provided by our randomized algorithm.
\end{example}

\begin{theorem} \label{thm:arbpat}
The online algorithm based on LP~\eqref{lp:arbpatience} is a $1/2$-approximation for the star graph probing problem, for an arbitrary patience distribution which is given explicitly.
\end{theorem}
The proof of \Cref{thm:arbpat} is in \Cref{sec:deferred_proofs}; we note that
the analysis in the proof is tight.

We further note that the result of \Cref{thm:arbpat} compares to a benchmark (LP~\eqref{lp:arbpatience}) that
does \emph{not} know the full realization of the patience values in advance. This is
necessary, since \Cref{thm:unknownpatience-badcr} states that comparing to a benchmark which knows the
patience in advance leads to arbitrarily bad competitive ratios.

\Cref{thm:arbpat} allows for us to solve online matching problems
when the patience of each customer is stochastic and follows an \emph{arbitrary}
distribution that is known to the algorithm. We simply use our star graph
algorithm as a black box for our algorithms in \Cref{sec:lpalgs}, with $\kappa=1/2$.
This gives us $\frac{1}{4}$-competitive algorithms for vertex-weighted adversarial
arrivals and edge-weighted prophet arrivals; as per
\Cref{thm:prophetiid,thm:prophetvw}, we have improved competitive ratios of
$\frac{1}{2}(1-1/\euler)$ under known IID arrivals (even with arbitrary edge weights)
and vertex-weighted prophet arrivals (even when the distributions are not
identical).

\subsection{Item Arrivals} \label{sec:item_arrivals}
Next we consider a different setting in which after a customer arrives, the "items" (interpreted as contractors in an online platform)
are initially unavailable, and only show up (to acknowledge that they can do the customer's job) following Bernoulli processes.
More specifically, each item $i \in [m]$ has two given
probabilities: the matching probability $p_{i}$ and an arrival probability $q_{i}$.
The customer has a known deterministic patience $\pat$, and the process unfolds in
discrete time steps; at time $t\in\{1, \dots, \pat\}$, the algorithm must choose
at most one item that has arrived to offer to the customer. After time $t=\pat$, if no
item has been purchased, the customer runs out of patience and becomes
permanently unavailable for matches. In contrast to the other patience settings,
in this setting the patience corresponds to the \emph{amount of time} the
customer is willing to wait, rather than the \emph{number of items} they may be
offered. Thus, if the algorithm makes no offer at time $t$ (because it is waiting for items to become available), we still move one
step closer to exhausting the customer's patience.

When the customer first arrives (immediately prior to time $t=1$), no items are
available to be offered. However, at each time step, each item $i$ becomes
available independently with probability $q_{i}$. Once an item $i$ becomes available to the customer, say initially at time $t$, then it can be offered to the customer at most once, at any time step $t' \ge t$. When item
$i$ is offered, the customer purchases it with probability $p_{i}$, in which
case a weight of $w_{i}$ is achieved and the process terminates; with
probability $1-p_{i}$, the item is not purchased and the process immediately
proceeds to the next time step. As with all star graph problems, our goal is to
develop an algorithm which maximizes the expected weight of the item sold to the
customer (achieving a weight of $0$ if no item is sold).

We begin with a linear programming relaxation of the problem.
\begin{subequations} \label{lp:itemarrivals}
\begin{alignat}{3}
\mathrm{LP} :=\max
& \sum_{i=1}^{m} w_{i}x_{i}p_{i} \tag{\ref{lp:itemarrivals}} \\[2ex]
\text{subject to }
& \sum_{i=1}^{m}x_{i}p_{i} \le 1 & & \label{eq:lp-1match} \\
& \sum_{i=1}^{m} x_{i} \le \pat & & \label{eq:lp-tprobes} \\
& x_{i} \le 1-(1 - q_{i})^{\pat} & &\hspace{3em}\forall i\in\{1,2,\dots,m\} \label{eq:lp-xibound} \\[2ex]
& x_{i} \ge 0 & &\hspace{3em}\forall i\in \{1,2,\dots,m\}
\end{alignat}
\end{subequations}

We call an item "large" if $q_{i} \ge c / \pat$. Otherwise, if
$q_{i} < c / \pat$, we say that item $u_{i}$ is "small." Let
$\LRG = \{i \in [m] \mid x_{i} \ge c/\pat \}$ denote the set of large items, and
$\SML = \{i \in [m] \mid x_{i} < c/\pat \}$ denote the set of small items. Our algorithm
first solves the LP~\eqref{lp:itemarrivals} to obtain an optimal solution $(x_{i})_{i\in[m]}$; then, it
makes use of one of two different strategies, choosing between the two depending
on relative contribution of large vs small items to the LP objective. The
motivation here is as follows: intuitively, we wish to choose to offer an item
$i$ with some probability proportional to $x_{i}$. However, if most of the
contribution to the objective value in the LP comes from small items, there may
be a high probability of no items arriving in any one time step; in this
case, we may be better off simply offering any item that arrives in a time step
where we are lucky enough to have an arrival.

\textbf{The \texttt{LARGE} strategy.}
First, let $\rho\in (0,\frac{1}{2})$ be a fixed parameter. Our strategy
$\pi_{\textsc{Large}}$ does the following: at each time step, we select an item
at random. When selecting a random item, we choose item $i$ with probability
$x_{i}/\theta$. With probability $1 - \sum_{i=1}^{m} x_{i}/\pat$, we select no item. It
follows from constraint~\eqref{eq:lp-tprobes} that this forms a valid
distribution.

The algorithm selects this item at random, and if the item has arrived and has
not yet been offered, we offer it to the buyer with probability $\rho$ (and with
probability $1-\rho$ we make no offer).

\textbf{The \texttt{SMALL} strategy.}
Our strategy $\pi_{\textsc{Small}}$ does the following: At each time step, if at
least one small item arrives in that step, choose one of the small arrivals at
random and offer it to the buyer. We ignore large items. Any small item which
arrives and is not chosen is permanently discarded (i.e., it will never be
offered to the buyer).

\textbf{The full algorithm.}
We fix a parameter $\varphi\in (0,1)$. First, if $\pat=1$, we simply take the optimal
choice, offering the item with the highest expected reward among items that
arrived. Then, for $\pat\ge 2$, if
$\sum_{i \in \LRG} w_{i}x_{i}p_{i} \ge (1-\varphi)\mathrm{LP}$, we use strategy
$\pi_{\textsc{Large}}$ at every time step. Otherwise,
$\sum_{i \in \SML} w_{i}x_{i}p_{i} > \varphi\mathrm{LP}$, and we use $\pi_{\textsc{Small}}$
at every time step.

We show that this algorithm is a $0.027$-approximation for the problem. This is
done by considering the two cases (corresponding to using the \texttt{LARGE} and \texttt{SMALL}
strategies) separately.

\begin{lemma} \label{lem:itemarrivals-large}
If $\sum_{i \in \LRG} w_{i}x_{i}p_{i} \ge (1-\varphi)\mathrm{LP}$, then
$\pi_{\textsc{Large}}$ achieves an expected matching weight of
$(1-\varphi)c\rho(1-2\rho)\mathrm{LP}$.
\end{lemma}

\begin{lemma}\label{lem:itemarrivals-small}
If $\sum_{i\in \LRG} w_{i}x_{i}p_{i} < (1-\varphi) \mathrm{LP}$, then
$\pi_{\textsc{Small}}$ achieves an expected matching weight of at least
\begin{equation}
\varphi \left(1 - \frac{c}{2}\right)^{\frac{2}{1 - (1-c/2)^{2}}} \left(\frac{1}{c} - \frac{e^{-c}}{1 - e^{-c}} \right)\mathrm{LP}
\end{equation}
\end{lemma}

Using the bounds for both the \texttt{LARGE} and \texttt{SMALL} strategies, we
can now give our final result.
\begin{theorem}\label{thm:itemarrivals}
For an appropriate choice of parameters $\varphi, \rho, c$, our algorithm is a
$0.027$-approximation.
\end{theorem}

\Cref{lem:itemarrivals-large,lem:itemarrivals-small} and \Cref{thm:itemarrivals} are proved in \Cref{sec:deferred_proofs}.
Our result for this setting can be used as a black box for a new kind of online stochastic
matching problem with two-sided arrivals, where after the arrival of each customer, all
items (contractors in an online labor platform) are initially unavailable, and arrive over time following Bernoulli processes (when they "discover" the customer's task). These acknowledgments "reset" after each customer, who presents a new job, and we note that that each (contractor, customer)-pair can have a different rate for its Bernoulli process of the contractor arriving, as well as a different probability for the customer accepting that contractor. A contractor, once matched, spends the time horizon (e.g.\ one day) doing that task and hence never returns.  For such an online
matching problem, we may use our strategy as a black box for the algorithms of
\Cref{sec:lpalgs}, where we have $\kappa=0.027$, to get a constant-factor competitive ratio for any finite market size and time horizon.

\section{Negative Results} \label{sec:neg_results}

\subsection{Stochasticity Gap} \label{sec:stochgap}
The \emph{stochasticity gap} is a fundamental gap in Linear Programming relaxations for stochastic problems which replace probabilities with deterministic fractional weights. The notion was first discussed informally in~\citet{brubach2017attenuate}, and was later also observed by~\citet{purohit2019hiring} (where they referred to it as a "probing gap"). When these LP relaxations are used as upper bounds on the offline optimal solution, or as a benchmark for the competitive ratio, the stochasticity gap represents a barrier to the best achievable competitive ratio. One interpretation of such a result is that better competitive ratios are not possible. However, one may alternatively view it as a result showing the limitations of using a particular LP as a benchmark for competitive ratios.

We present a stochasticity gap for a common LP relaxation of the online matching
problem with stochastic rewards. Recall from \Cref{sec:advlp} that LP~\eqref{lp:adv} with
the constraints~\eqref{eq:lpadv-1matchu}--\eqref{eq:lpadv-tprobes} (and excluding our additional family of
constraints~\eqref{eq:lpadv-dpopt}) is a standard LP relaxation for bipartite matching with (known)
patience constraints and adversarial arrivals. This is essentially an extension
of the "Budgeted Allocation" LP from~\citet{mehta2012online} to include the patience constraints.
For convenience, we reproduce this standard LP below in the vertex-weighted
setting.
% \begin{subequations} \label{lp:stochgap}
\begin{alignat*}{3}
\max &\sum_{u\in U} \sum_{v\in V}
x_{u,v}p_{u,v} w_{u} \\[2ex] %\tag{\ref{lp:stochgap}} \\[2ex]
\text{subject to }
& \sum_{v\in V} x_{u,v}p_{u,v} \le 1 & &\hspace{2em}\forall u\in U \\
& \sum_{u\in U} x_{u,v} p_{u,v} \le 1 & &\hspace{2em}\forall v\in V \\
& \sum_{u\in U} x_{u,v} \le \bE[\pat_{v}] & &\hspace{2em}\forall v\in V \\
& 0 \le x_{u,v} \le 1 & &\hspace{2em}\forall u\in U, v\in V \notag
\end{alignat*}
% \end{subequations}

Simple LP formulations like this, while useful, can give too large of an upper
bound on the performance of any offline algorithm, and thus make it difficult to
get larger competitive ratios. As such, more complex (and, often,
exponentially-sized) LPs have been used in recent work (see, e.g.,~\citet{svenssonLP}) to
achieve better results. Our LP-based techniques in \Cref{sec:lpalgs} use different
exponential-sized LPs to overcome the limitations of stochasticity gaps.

We start with a simple example demonstrating the notion of stochasticity gap,
where the bipartite graph has a single offline vertex $u$, and $n$ online
vertices arriving in any order. Suppose $p_{uv} = 1/n$ and $w_{uv}=1$ for all
online vertices $v\in V$. The LP given by~\eqref{eq:lpadv-1matchu}--\eqref{eq:lpadv-tprobes} can assign $x_{uv}=1$ for all
edges, achieving an objective value of $1$. However, the best any online
algorithm can do is probe the single edge $(u,v)$ whenever vertex $v$ arrives
online, which matches the single offline vertex $u$ with probability $1-1/e$.
Thus, it is impossible for any online algorithm to guarantee a matching of
expected weight better than $(1-1/e)$ times the LP value. This establishes a
stochasticity gap of $1-1/e$ for this formulation, and suggests that if we wish
to beat the $1-1/e$ barrier, we must use a different LP benchmark. However, the
stochasticity gap for the LP of~\eqref{eq:lpadv-1matchu}--\eqref{eq:lpadv-tprobes} is even worse. To establish this,
  we consider a complete bipartite graph with $n$ vertices on each side, and
  edge probabilities $1/n$; a result on random graphs
  then implies~\Cref{thm:mehtalpstochgap}, whose proof is in
  \Cref{sec:deferred_proofs}.
% . first cite a result about random graphs due
% to~\citet{bb95randgraph}.
% \begin{lemma}[Theorem~14 of~\citet{bb95randgraph}]  \label{lem:randgraph}
% Let $G$ be a random bipartite graph with both partitions of size $n$ and where each edge exists independently with probability $p=1/n$. Let $\gamma$ be the solution to the equation $\gamma = e^{-\gamma}$. Then, the largest independent set of $G$ has size $n(2\gamma+\gamma^{2}))[1+o(1)]$ with probability $1-o(1)$.
% \end{lemma}

% This result allows us to derive the following stochasticity gap, whose proof is in \Cref{sec:deferred_proofs}.
\begin{theorem} \label{thm:mehtalpstochgap}
The LP given by the objective function~\eqref{lp:adv} and constraints~\eqref{eq:lpadv-1matchu}--\eqref{eq:lpadv-tprobes} has a stochasticity gap of at most $\approx 0.544$.
\end{theorem}

We should note that \cite{fata2019multi} establishes a smaller upper bound of $1-\ln(2-1/e)\approx0.51$ relative to this LP, but they restrict to online probing algorithms.  Our higher upper bound holds even for the offline optimal matching, hence reflecting a true "stochasticity gap".

\subsection{$0.5$ Upper Bound for SimpleGreedy} \label{sec:simple_greedy}
As defined in~\citet{mehta2012online}, an \emph{opportunistic}
algorithm for the \emph{Stochastic Rewards} setting is
one which always attempts to probe an edge incident to an
online arriving vertex $v\in V$ if one exists.
The work of~\citet{mehta2012online} showed that in the
unweighted \emph{Stochastic Rewards} ($\pat_v=1$ for all online vertices $v\in V$) problem, any
opportunistic algorithm achieves a competitive ratio of
$1/2$.
The simplest opportunistic algorithm is the one which, when
$v\in V$ arrives online, chooses a neighbor $u\in U$ of $v$
arbitrarily and probes the edge $(u,v)$. We call this algorithm
"SimpleGreedy". Since SimpleGreedy is opportunistic, the result of~\citet{mehta2012online} shows that
SimpleGreedy achieves a competitive ratio of at least $1/2$; \Cref{thm:halfub} shows that this is tight even when restricted to small, uniform $p$.
\begin{theorem} \label{thm:halfub}
There exists a family of unweighted graphs under stochastic rewards and adversarial arrivals for which \textup{SimpleGreedy} achieves a competitive ratio of at most $1/2$ even when all edges have uniform probability $p = O(1/n)$.
\end{theorem}

We present our construction here. Let $k$ be a fixed positive integer constant. Let
$U=U_{0}\cup U_{n}$, where $U_{0}$ and
$U_{n} = \{u_{1},\dots, u_{n}\}$ are disjoint, and $|U_{0}|=k$. Let
$V = V_{0}\cup V_{n}$ where $V_{0}$ and
$V_{n}=\{v_{1},\dots,v_{n}\}$ are disjoint, and
$|V_{0}|=kn^{2}$. Let $E=E_{0}\cup E_{n}$ where
$E_{0}= U_{0}\times V$ and
$E_{n} = \{(u_{i},v_{i}) \mid i=1,\ldots,n\}$.
Let $p=k/n$.

For the bipartite graph $G(U,V; E)$, an offline algorithm can
achieve a matching of expected size at least $2k$ by first probing
edges $(u,v)\in U_{0}\times V_{0}$ until all edges are probed or the
maximum possible successful matches, $k$, is achieved. This strategy
achieves $k$ successful matches among these edges in
expectation. Then, the offline optimal will probe all edges of
$E_{n}$ in any order, achieving an expected number of successful
matches of $k$. The total expected size of the achieved matching is
then $2k$.
We complete the proof of \Cref{thm:halfub} by showing that an online algorithm cannot earn more than $k+o(k)$, in \Cref{sec:deferred_proofs}.

\subsection{Hardness of Unknown Patience} \label{sec:unknown_pat}
We now show that when offering items to a single customer with random patience, one should not be comparing to a benchmark that knows the realization of the patience in advance, or else the competitive ratio will be 0.
The same counterexample shows for single-item IID-valued online accept/reject problems that the competitive ratio will be 0 if the number of arrivals is unknown, recovering the result of \citet{alijani2020predict}.

\begin{theorem} \label{thm:unknownpatience-badcr}
For the star graph probing problem with patience $\pat$ drawn from an arbitrary distribution\footnote{See \Cref{sec:arbpatience} for the exact problem statement.  There we compared to an LP that also did not know the realization of the patience in advance, the importance of which is fully justified by the present \namecref{thm:unknownpatience-badcr}.}, the reward of an online algorithm relative to a clairvoyant who sees the realization of $\pat$ in advance must be 0, even if there are infinite copies of every offline vertex.
\end{theorem}

\Cref{thm:unknownpatience-badcr} is proved in \Cref{sec:deferred_proofs}, and its construction is presented below.
The significance of allowing infinite copies of offline vertices is the following.
Essentially, we are left with a pricing problem where there are an unknown $\pat$ number of opportunities to make a single sale to a customer; the different offline vertices' weights correspond to different prices that can be tried, and after each trial we get an independent realization (due to the infinite copies) whose probability depend on the price.
One can further transform such an instance into an online accept/reject problem facing a stream of $\pat$ IID draws, where the pricing decisions correspond to acceptance thresholds.
Therefore, our hardness result implies the following.  Although already known to \citet[Appendix~A.1]{alijani2020predict}, we re-derive it using our construction to articulate the connection, which we believe is instructive.

\begin{corollary} \label{cor:unknownpatience}
Consider the simple optimal stopping problem where an online algorithm can accept at most one of $\theta$ values that are drawn IID from a known distribution and presented one-by-one.
If $\theta$ is unknown, then the competitive ratio is 0.
\end{corollary}

% The optimal stopping problem described in \Cref{cor:unknownpatience} is equivalent to the \textit{IID prophet}
% problem \citep{correa2017posted}, and the \textit{single-leg revenue management} problem \citep{gallego1994optimal} with one unit of
% inventory. In those problems, the number of buyers or horizon length is always
% assumed to be known, leading to constant-factor guarantees. In
% stark contrast, if the number of arrivals is unknown, then a
%   constant factor is impossible.

\textbf{Our construction.}
Fix a positive integer $k$ and let $m$ be another positive integer that we will drive to $\infty$.
Consider a star graph, i.e.\ a bipartite graph with many offline vertices $u\in U$
and a single online vertex $v$. Consider the following distribution over the
patience of $v$:
\begin{equation}
\label{eq:badpatdist}
\tilde{\pat}_v =
\begin{cases}
m^{2i} & w.p.~m^{-i}-m^{-i-1} \qquad \forall i=0,\ldots,k-1; \\
m^{2k} & w.p.~m^{-k}.  \\
\end{cases}
\end{equation}
In our construction, there are $m^{2k}$ identical offline vertices for each $i=0,\ldots,k$, with weight $m^i$ and probability $m^{-2i}$.  We note that $m^{2k}$ is greater than the largest possible realization of $\tilde{\pat}_v$, so the constraints on the availability of offline vertices are never binding.  That is, \textit{our construction applies even in the more restrictive setting} where there are infinite copies of every offline vertex.

\textbf{Intuition behind hardness.}
In our construction, there are essentially an unknown number of opportunities to sell a single item.
During each opportunity, one must choose a consumption option $i=0,\ldots,k$, which has a $m^{-2i}$ probability of selling the item at price $m^i$.
The immediate reward from consumption option $i$ is $m^i\cdot m^{-2i}=m^{-i}$, which is decreasing in $i$, but smaller indices of $i$ also has a higher chance of closing the sale and eliminating future opportunities.
Therefore, there is a tradeoff between offering "longshot" prices with a high index of $i$ (desirable if a large number of opportunities remain), vs.\ the "safe" option $i=0$ which makes a sale w.p.~$m^{-0}=1$ (desirable on the final opportunity).
Our proof of \Cref{thm:unknownpatience-badcr} in \Cref{sec:deferred_proofs} shows that a clairvoyant who tries only option $i$ when they know the patience will be $m^{2i}$, for all $i=0,\ldots,k$, can earn $\approx(1-1/e)(k+1)$.
Meanwhile, any online algorithm is best off using the "safe" option $i=0$ on the first try and finishing, since the customer only has a small chance of having patience greater than 1 (we prove this through backwards induction on the optimal dynamic program).
This establishes an unbounded separation when $k$ is taken to be large in our construction.

\textbf{Transformed hard instance to establish \Cref{cor:unknownpatience}.}
For concreteness, we show how to transform our construction to the accept/reject problem.
The IID draws from the distribution should take one of $k+1$ possible values, indexed by $i=0,\ldots,k$.
The decision each period is to set a threshold on the minimum acceptable value, where each option $i=0,\ldots,k$ should correspond to a "consumption option" that has probability $m^{-2i}$ of accepting.
Therefore, the probability of the IID draw taking value index $i$ for all $i=0,\ldots,k-1$ should be $m^{-2i}-m^{-2(i+1)}$ and the probability of value index $i=k$ should be $m^{-2k}$, so that accepting all levels with index at least $i$ has probability $$(m^{-2i}-m^{-2(i+1)})+(m^{-2(i+1)}-m^{-2(i+2)})+\cdots+m^{-2k}=m^{-2i}$$ of making an acceptance.
Now, the exact values for each index $i$ must be calibrated so that the immediate reward from each consumption option $i$ is $m^{-i}$.
Using backwards induction over $i=k,\ldots,1$, we can solve that the value for index $i=k$ should be $m^k$ and that the value for each index $i=0,\ldots,k-1$ should be
$$\frac{m^{-i}-m^{-(i+1)}}{m^{-2i}-m^{-2(i+1)}}=\frac{m^{i+2}-m^{i+1}}{m^2-1}=\frac{m^{i+1}}{m+1}.$$
This completes the construction of our transformed instance for the accept/reject problem.

\textbf{Applying Yao's minimax principle.}
Finally, we explain why in \Cref{cor:unknownpatience}, there is no difference between $\theta$ being completely unknown vs.\ $\theta$ being drawn from a known distribution (but competing against a clairvoyant who knows its realization in advance).
Formally, for any patience $\pat$ and any (deterministic) non-clairvoyant algorithm $\psi$, let $\ALG(\psi,\pat)$ denote the algorithm's expected reward when the patience realizes to $\pat$.  Meanwhile, let $\OPT(\pat)$ denote the clairvoyant's expected reward when the patience is known to be $\pat$.
Let $D$ denote a distribution over patiences $\pat$, and $\Psi$ denote a distribution over algorithms $\psi$.
Yao's minimax principle says that
\begin{align*}
\sup_\Psi \inf_\pat \frac{\bE_{\psi\sim\Psi}[\ALG(\psi,\pat)]}{\OPT(\pat)}
= \inf_D\sup_{\psi}\bE_{\pat\sim D}\left[\frac{\ALG(\psi,\pat)}{\OPT(\pat)}\right]
=\inf_D\frac{\sup_{\psi}\bE_{\pat\sim D}[\ALG(\psi,\pat)]}{\bE_{\pat\sim D}[\OPT(\pat)]}
\end{align*}
(where the second equality holds via rescaling worst-case distributions $D$ by $\OPT(\theta)$).
The existence of our family of distributions in~\eqref{eq:badpatdist} shows that the RHS expression, and hence all of these expressions, equal 0.
The LHS expression equaling 0 implies that for any fixed (randomized) online algorithm that does not know the value of $\theta$ in advance, an adversary can always set a horizon length $\theta$ for which the algorithm performs unboundedly worse relative to $\OPT(\theta)$.

% \section{Conclusion and Future Directions}
% Our results provide improved bounds for several well-known online matching problems, specifically ones that involve stochastic rewards and patience parameters. We further introduce a new repeated-interaction model for a single customer who has a patience that is randomly determined and initially unknown, which can open up many avenues for future work.

\ACKNOWLEDGMENT{
Nathaniel Grammel was supported in part by NSF award CCF-1918749, and by research awards from Amazon and Google.

Will Ma would like to thank Jake Feldman and Danny Segev for instructive discussions relevant to this paper.

Aravind Srinivasan was supported in part by NSF awards CCF-1422569, CCF-1749864, and CCF-1918749, as well as research awards from Adobe, Amazon, and Google.
}

\bibliographystyle{informs2014} % outcomment this and next line in Case 1
\bibliography{bibliography} % if more than one, comma separated

% FOR SUBMISSION
% \ECSwitch
% \ECDisclaimer
% \ECHead{E-Companion}

% FOR SSRN
\clearpage

% Appendix here
% Options are (1) APPENDIX (with or without general title) or
%             (2) APPENDICES (if it has more than one unrelated sections)
% Outcomment the appropriate case if necessary
%
% \begin{APPENDIX}{<Title of the Appendix>}
% \end{APPENDIX}
%
%   or
%

\begin{APPENDICES}
  \crefalias{section}{appendix}

\section{Deferred Proofs} \label{sec:deferred_proofs}

\begin{proof}{Proof of \Cref{lem:lpopt}.}
Consider an adaptive offline algorithm which is optimal. Let $x_{u,v}$ be the probability that this strategy probes edge $(u,v)$. For any vertex $u\in U$, the probability that $u$ is successfully matched is at most $\sum_{v\in V}x_{u,v}p_{u,v}\le 1$, and similarly for the probability of successfully matching any online vertex $v\in V$. Thus, this assignment satisfies constraints~\eqref{eq:lpadv-1matchu} and~\eqref{eq:lpadv-1matchv}. By the definition of $\OPT(G)$, we cannot probe more than $\pat_{v}$ edges incident on an online vertex $v$. So constraint~\eqref{eq:lpadv-tprobes} is satisfied.

Finally, we argue that the new constraint~\eqref{eq:lpadv-dpopt} is satisfied by this assignment.
Suppose instead there is some vertex $v'\in V$ and some $U'\subseteq U$ for which
$\sum_{u\in U'}x_{u,v'}p_{u,v'}w_{u} > \OPT(U',v')$. Then, we can define a new
offline probing strategy on the star graph $(U',\set{v'}, U'\times\set{v'})$ which
simply simulates our original algorithm on $G$ and probes only those edges which
are in $U'\times\set{v'}$. This achieves an expected matching weight on the star
graph of at least $\sum_{u\in U'}x_{u,v'}p_{u,v'}w_{u} > \OPT(U',v')$, but this
contradicts the fact that $\OPT(U',v')$ is the optimal expected matching weight
for the star graph. Thus, this assignment must satisfy constraint~\eqref{eq:lpadv-dpopt}. It
follows that LP~\eqref{lp:adv} must have objective value at least as large as
the expected matching weight of the optimal offline algorithm. \Halmos
\end{proof}

\begin{proof}{Proof of \Cref{lem:halflp}.}
For the sake of analysis, suppose we have solved LP~\eqref{lp:adv} on the graph $G$. Let
\[c(x) = \sum_{u\in U}\sum_{v\in V}x_{u,v}p_{u,v}w_{u}\]
be the value of the objective function for $x$ and let
\[c_v(x) = \sum_{u \in U} x_{u,v} p_{u,v} w_u\]
be the value ``achieved'' by a given online vertex $v$ with $c(x) = \sum_{v \in V} c_v(x)$. Let $x^{*}$ be the optimal assignment given by LP~\eqref{lp:adv}, for the graph $G$. So $c(x^{*}) = \sum_{v \in V} c_v(x^*) = \LPOPT(G)$.

We will make the following charging argument. Imagine that when a vertex $v$ is matched to some $u \in U$, we assign $0.5 w_u$ to $v$ and for all $v' \in V$ (including $v$ itself) we assign $0.5 x_{u,v'} p_{u,v'} w_u$ to $v'$. Note we have assigned at most $w_u$ weight in total since $\sum_{v' \in V} 0.5 x_{u,v'} p_{u,v'} w_u \leq 0.5 w_u$ due to LP constraint~\eqref{eq:lpadv-1matchu}.

Let $w_{v}$ for online vertex $v \in V$ be equal to the weight $w_u$ of the offline vertex $u$ which is matched to $v$ or $0$ if $v$ is unmatched at the end of the arrivals. Let $U_m \subseteq U$ be the set of offline vertices which are matched at the end of the arrivals. We define
\[c_v(\ALG) = 0.5 w_{v} + 0.5 \sum_{u \in U_m} x_{u,v} p_{u,v} w_u\]
as the weight assigned to $v$ in our imaginary assignment. By the linearity of expectation
\[
\ALG(G) = \bE\left[\sum_{v \in V} c_v(\ALG)\right] = \sum_{v \in V} \bE[c_v(\ALG)]
\]
Thus, to complete the proof, we must show that
$$\sum_{v \in V} \bE[c_v(\ALG)] \geq \frac{1}{2}\LPOPT(G)$$

Consider an online vertex $v$ arriving at time $k$. Let $U_{v}\subseteq U$ be the set of vertices available (unmatched) when $v$ arrives and $U_{-v} = U\setminus U_{v}$ be the set of vertices which are already matched when $v$ arrives. Note that when $v$ arrives, it has already been assigned a value of $0.5 \sum_{u \in U_{-v}} x_{u,v} p_{u,v} w_u$.
After using a $\kappa$-approximate ($\kappa \le 1$) black box for matching $v$ to $U_v$, we have assigned an expected value to $v$ of at least
$0.5\kappa \OPT(U_v, v) + 0.5 \sum_{u \in U_{-v}} x_{u,v} p_{u,v} w_u$.

Thus, we have
\begin{align*}
\sum_{v \in V} \bE[c_v(\ALG)]
&\geq \sum_{v \in V}
\sum_{U_v \subseteq U} \Pr[U_v]
\left( 0.5 \kappa\OPT(U_v, v) + 0.5 \sum_{u \in U_{-v}} x_{u,v} p_{u,v} w_u \right) \\
&\geq \sum_{v \in V}
\sum_{U_v \subseteq U} \Pr[U_v]
\left( 0.5 \kappa \sum_{u\in U_v} x_{u,v}p_{u,v}w_{u} + 0.5 \kappa \sum_{u \in U_{-v}} x_{u,v} p_{u,v} w_u \right)
\quad \because~\eqref{eq:lpadv-dpopt} \\
&\geq 0.5 \kappa \sum_{v \in V}
\sum_{U_v \subseteq U} \Pr[U_v]
\sum_{u\in U} x_{u,v}p_{u,v}w_{u} \\
&= 0.5 \kappa \sum_{v \in V} \sum_{u\in U} x_{u,v}p_{u,v}w_{u} \\
&=  \frac{\kappa}{2}\LPOPT(G)
\end{align*}
which completes the proof.\Halmos
\end{proof}
%%% Local Variables:
%%% mode: latex
%%% TeX-master: "follow_your_star_OR"
%%% End:

\begin{proof}{Proof of \Cref{thm:prophethalf}.}
Given a feasible solution $\ssx_v(\pi)$ to LP~\eqref{lp:prophet} whose objective value 
is at least $\kappa\cdot\OPTP(G)$, we show that an online algorithm can form a matching whose weight is in expectation at least $\frac{\kappa}{2}\cdot\OPTP(G)$.

First, we establish some notation. Let $Q^t_v\in\{0,1\}$ be the indicator random variable for the event that the arrival at time $t$ has type $v\in V$. 
Let $Y^t_{uv}\in\{0,1\}$ be the indicator random variable for $u$ being matched to a vertex of type $v$ at time $t$. Let $N^t_u$ be the indicator random variable for offline vertex $u$ becoming matched by the \textit{end} of time $t$. And, let $\Pi_t$ be denote the policy chosen by \Cref{alg:prophet} for time $t$. We may now write the online algorithm's weight matched as
\begin{equation} \label{eq:prophetALG}
\begin{aligned}
\sum_{t,u,v}w_{uv}Y^t_{uv}
&=\sum_{t,u,v}(w_{uv}-\frac{\ssw_u}{2})Y^t_{uv}+\sum_u\frac{\ssw_u}{2}\sum_{t,v}Y^t_{uv} \\
&=\sum_{u,v:w_{uv}\ge\ssw_u/2}(w_{uv}-\frac{\ssw_u}{2})\sum_tY^t_{uv}+\sum_u\frac{\ssw_u}{2}N^T_u
\end{aligned}
\end{equation}

Now we wish to analyze the value of $\bE[Y^t_{uv}]$. For a pair $u,v$ such that $w_{uv}\ge\ssw_u/2$ (i.e.\ edge $uv$ is not skipped over when executing policies $\pi\in\cP$ during the execution of \Cref{alg:prophet}), we have
\begin{equation} \label{eq:expYprophet}
\begin{aligned}
\bE[Y^t_{uv}]
&=\bE[Y^t_{uv}|(Q^t_v=1)\cap(N^{t-1}_u=0)]\Pr[(Q^t_v=1)\cap(N^{t-1}_u=0)] \\
&=q_{tv}(1-\bE[N^{t-1}_u])\sum_{\pi\in\cP}\bE[Y^t_{uv}|(Q^t_v=1)\cap(N^{t-1}_u=0)\cap(\Pi_t=\pi)]\Pr[\Pi_t=\pi|Q^t_v=1]
\end{aligned}
\end{equation}
Next, notice that $\bE[N^{t-1}_u]\le\bE[N^T_u]$ since the probability of offline vertices becoming matched only increases over time. Notice also that since $w_{uv}\ge\ssw_u/2$, no policy for a vertex of type $v$ will skip over probing $u$. Further, skipping over other vertices can only increase the probability that $u$ is probed. Hence, we can conclude that $\bE[Y^t_{uv}|(Q^t_v=1)\cap(N^{t-1}_u=0)\cap(\Pi_t=\pi)] \ge p_{uv}(\pi)$. Finally, we observe that $\Pr[\Pi_t = \pi | Q^t_v=1] = \ssx_v(\pi)/q_v$.

Applying the above to~\eqref{eq:expYprophet}, we have
\begin{align*}
\bE[Y^t_{uv}]
&\ge q_{tv}(1-\bE[N^T_u])\sum_{\pi\in\cP}p_{uv}(\pi)\frac{\ssx_v(\pi)}{q_v}
\end{align*}
and combining with~\eqref{eq:prophetALG} we get can bound the algorithm's expected weight as
\begin{align*}
\bE\left[\sum_{t,u,v}w_{uv}Y^t_{uv}\right]
&\ge\sum_{u,v:w_{uv}\ge\ssw_u/2}(w_{uv}-\frac{\ssw_u}{2})\underbrace{\sum_tq_{tv}}_\text{$=q_v$ by definition}(1-\bE[N^T_u])\sum_{\pi\in\cP}p_{uv}(\pi)\frac{\ssx_v(\pi)}{q_v}+\sum_u\frac{\ssw_u}{2}\bE[N^T_u] \\
&\ge\sum_u(1-\bE[N^T_u])\sum_v(w_{uv}-\frac{\ssw_u}{2})\sum_{\pi\in\cP}p_{uv}(\pi)\ssx_v(\pi)+\sum_u\frac{\ssw_u}{2}\bE[N^T_u] \\
&=\sum_u(1-\bE[N^T_u])\left(\ssw_u-\frac{\ssw_u}{2}\underbrace{\sum_v\sum_{\pi\in\cP}p_{uv}(\pi)\ssx_v(\pi)}_\text{$\le1$ by LP constraint~\eqref{eq:prophetlp-offline}}\right)+\sum_u\frac{\ssw_u}{2}\bE[N^T_u]
\end{align*}
which is at least $\sum_u\ssw_u/2$.  By assumption on the LP solution, we have $\sum_u\ssw_u/2\ge\kappa\cdot\OPTP(G)/2$, completing the proof.
\Halmos\end{proof}

\begin{proof}{Proof of \Cref{thm:prophetiid}.}
Given a feasible solution $\ssx_v(\pi)$ to LP~\eqref{lp:prophet} whose objective value 
is at least $\kappa\cdot\OPTP(G)$, we show that an online algorithm can form a matching whose weight is in expectation at least $\kappa(1-1/e)\cdot\OPTP(G)$.

We let $Q^t_v$, $Y^t_{uv}$, $N^t_u$, and $\Pi_t$ be defined as in the proof of \Cref{thm:prophethalf}. Further, let $Z^t_{uv}\in\{0,1\}$ be the indicator random variable for $(u,v)$ being \emph{successfully} probed (real or simulated) at time $t$. Then, let $Z^t := \sum_{v\in V} Z^t_{uv}$. Note that $Z^t\in\set{0,1}$ since $Z^t_{uv}$ can only be $1$ for the single $v\in V$ that corresponds to the vertex type arriving at time $t$.

Next, since $u$ is available to be matched at time $t$ if and only if $Z^1_u = Z^2_u = \dots = Z^{t-1}_u = 0$, we can write $Y^t_{uv} = Z^t_{uv} \cdot \bI[Z^1_u = Z^2_u = \dots = Z^{t-1}_u = 0]$.

Using linearity of expectation and independence, we can decompose the algorithm's weight as:
\begin{align} \label{eqn::original}
\bE\left[\sum_{t,u,v}w_{uv}Y^t_{uv}\right]
&=\sum_{u,t}\bE\left[\sum_vw_{uv}Z^t_{uv}\right]\bE[\bI(Z^1_u=\cdots=Z^{t-1}_u=0)]
\end{align}
Under IID assumption, we know $q_{tv}=q_v/T$ for all $t$, and we can derive the following:
\begin{align}
\bE[Z^t_{uv}]
&=\Pr[Q^t_v=1]\sum_{\pi\in\cP}\bE[Z^t_{uv}|(Q^t_v=1)\cap(\Pi_t=\pi)]\Pr[\Pi_t=\pi|Q^t_v=1] \nonumber \\
&=\frac{q_v}{T}\sum_{\pi\in\cP}p_{uv}(\pi)\frac{\ssx_v(\pi)}{q_v} \nonumber \\
&=\frac{1}{T}\sum_{\pi\in\cP}p_{uv}(\pi)\ssx_v(\pi) \label{eqn::intermediate}
\end{align}
Using~\eqref{eqn::intermediate}, we can derive the following for any offline vertex $u$:
\begin{align}
\bE\Big[\sum_vw_{uv}Z^t_{uv}\Big|Z^t_u=1\Big]
&=\frac{\sum_vw_{uv}\bE[Z^t_{uv}]}{\Pr[Z^t_u=1]} \nonumber \\
&=\frac{\sum_vw_{uv}\bE[Z^t_{uv}]}{\sum_v\bE[Z^t_{uv}]} \nonumber \\
&=\frac{\sum_{v,\pi}w_{uv}p_{uv}(\pi)\ssx_v(\pi)}{\sum_{v,\pi}p_{uv}(\pi)\ssx_v(\pi)} \label{eqn::toUse}
\end{align}
Substituting~\eqref{eqn::toUse} back into~\eqref{eqn::original}, we get
\begin{align}
\bE\left[\sum_{t,u,v}w_{uv}Y^t_{uv}\right]
&=\sum_{u,t}\left(\frac{\sum_{v,\pi}w_{uv}p_{uv}(\pi)\ssx_v(\pi)}{\sum_{v,\pi}p_{uv}(\pi)\ssx_v(\pi)}\right)\Pr[Z^t_u=1]\bE[\bI(Z^1_u=\cdots=Z^{t-1}_u=0)] \nonumber \\
&=\sum_u\left(\frac{\sum_{v,\pi}w_{uv}p_{uv}(\pi)\ssx_v(\pi)}{\sum_{v,\pi}p_{uv}(\pi)\ssx_v(\pi)}\right)\sum_t\bE[\bI(Z^t_u=1)\bI(Z^1_u=\cdots=Z^{t-1}_u=0)] \label{eqn::finish}
\end{align}
For any $u$, let's analyze $\bE[\sum_t\bI(Z^t_u=1)\bI(Z^1_u=\cdots=Z^{t-1}_u=0)]$.  The expression inside the $\bE$ is 0 if all of the $Z^t_u$ are 0, and 1 otherwise.  Therefore, by independence, the expectation equals
\begin{align*}
1-\prod_t\Pr[Z^t_u=0]=1-\prod_t(1-\bE[\sum_vZ^t_{uv}]).
\end{align*}
Substituting in~\eqref{eqn::intermediate}, we get this is equal to
\begin{align*}
1-\prod_t(1-\sum_v\frac{1}{T}\sum_{\pi\in\cP}p_{uv}(\pi)\ssx_v(\pi))
&\ge1-\exp(-\sum_t\frac{1}{T}\sum_{v,\pi}p_{uv}(\pi)\ssx_v(\pi)) \\
&=1-\exp(-\sum_{v,\pi}p_{uv}(\pi)\ssx_v(\pi)) \\
\end{align*}
Finally, substituting back into~\eqref{eqn::finish}, we get that the algorithm's expected weight is at least
\begin{align*}
\sum_u\sum_{v,\pi}w_{uv}p_{uv}(\pi)\ssx_v(\pi)\frac{1-\exp(-\sum_{v,\pi}p_{uv}(\pi)\ssx_v(\pi))}{\sum_{v,\pi}p_{uv}(\pi)\ssx_v(\pi)}
\end{align*}
The function $\frac{1-\exp(-z)}{z}$ is decreasing in $z$ and the expression $\sum_{v,\pi}p_{uv}(\pi)\ssx_v(\pi)$ is at most 1 by LP constraint~\eqref{eq:prophetlp-offline}, hence the fraction in the expression above is uniformly lower-bounded by $1-1/e$.
Therefore, the algorithm's expected weight is at least $(1-1/e)\sum_u\sum_{v,\pi}w_{uv}p_{uv}(\pi)\ssx_v(\pi)$, which in turn is at least $(1-1/e)\kappa\cdot\OPT(G)$ by presumption on the given LP solution $\ssx_v(\pi)$, completing the proof.
\Halmos\end{proof}

\begin{proof}{Proof of \Cref{thm:prophetvw}.}
Given a feasible solution $\ssx_v(\pi)$ to LP~\eqref{lp:prophet} whose objective value 
is at least $\kappa\cdot\OPTP(G)$, we show that an online algorithm can form a matching whose weight is in expectation at least $\kappa(1-1/e)\cdot\OPTP(G)$.

We let $Q^t_v$, $Y^t_{uv}$, $Z^t_{uv}$, $Z^t_u$, and $N^t_u$ be as defined in the proof of \Cref{thm:prophethalf}.
The LP objective value can be written as 
\[
\sum_{u\in U}w_u\left(\sum_{v\in V}\sum_{\pi\in\cP}\ssx_v(\pi)p_{uv}(\pi)\right).
\]

We will show that \Cref{alg:prophetiidvw} in this case matches each offline vertex $u$ with probability at least $(1-1/e)\sum_{v\in V}\sum_{\pi\in\cP}\ssx_v(\pi)p_{uv}(\pi)$. Notice that
a vertex $u$ is left unmatched only if $Z^1_u=\cdots=Z^T_u=0$.  Since $Z^1_u,\ldots,Z^T_u$ are independent, the probability of this is
\begin{align*}
\prod_{t=1}^T(1-\Pr[Z^t_u=1])
&=\prod_{t=1}^T\left(1-\sum_{v\in V}\bE[Z^t_{uv}]\right) \\
&=\prod_{t=1}^T\left(1-\sum_{v\in V}q_{tv}\sum_{\pi\in\cP}p_{uv}(\pi)\frac{\ssx_v(\pi)}{q_v}\right) \\
&\le\prod_{t=1}^T\exp\left(-\sum_{v\in V}q_{tv}\sum_{\pi\in\cP}p_{uv}(\pi)\frac{\ssx_v(\pi)}{q_v}\right) \\
&=\exp\left(-\sum_{v\in V}\sum_{\pi\in\cP}p_{uv}(\pi)\ssx_v(\pi)\sum_{t=1}^T\frac{q_{tv}}{q_v}\right) \\
&=\exp\left(-\sum_{v\in V}\sum_{\pi\in\cP}p_{uv}(\pi)\ssx_v(\pi)\right)
\end{align*}
Therefore, the probability of \Cref{alg:prophetiidvw} matching an offline vertex $u$ is at least
$$1-\exp(-\sum_{v\in V}\sum_{\pi\in\cP}p_{uv}(\pi)\ssx_v(\pi))$$
which we know is at least $(1-1/e)\sum_{v\in V}\sum_{\pi\in\cP}\ssx_v(\pi)p_{uv}(\pi)$ by the same analysis as at the end of the proof of \Cref{thm:prophetiid}.
Therefore, the total expected weight accrued by the algorithm across all offline vertices $u\in U$ is at least $(1-1/e)\sum_{u\in U}\sum_{v\in V}\sum_{\pi\in\cP}\ssx_v(\pi)p_{uv}(\pi)$, which in turn is at least $(1-1/e)\kappa\cdot\OPTP(G)$ by presumption on the given LP solution $\ssx_v(\pi)$, completing the proof.
\Halmos\end{proof}

\begin{proof}{Proof of \Cref{thm:arbpat}.}
In our analysis, we define $S_{\pat}$ to be the random indicator
variable corresponding to the event that the probing process has
continued to the $\pat^{\text{th}}$ attempt. Additionally, $X_{j\pat}$
is the indicator variable for the event that vertex $j$ is probed in
the $\pat^{\text{th}}$ attempt, \emph{including simulated probes}.

We claim that $\bE[S_{\pat}] = s^{*}_{\pat}$ for all $\pat$. This can
be seen to be true by definition in~\eqref{eq:sUpdateOne}--\eqref{eq:sUpdateTwo} and the
fact that we simulate probes even when a vertex $j$ is unavailable.
Further, we can immediately see that
$\bE[X_{j\pat}] = \bE[X_{j\pat} | S_{\pat}=1]\Pr[S_{\pat}=1] =
\frac{x^{*}_{j\pat}}{s^{*}_{\pat}} s^{*}_{\pat} = x^{*}_{j\pat}$ as long as $\sss_\pat>0$, and that both $\bE[X_{j\pat}]$ and $\ssx_{j\pat}$ must be 0 if $\sss_\pat=0$.  Therefore, $\bE[X_{j\pat}]=\ssx_{j\pat}$ for all $j$ and $\pat$.

The algorithm's expected weight matched only counts non-simulated probes, and can be written
\begin{align} \label{eq:arbalg}
\sum_{j=1}^n w_jp_j\sum_{\pat=1}^n\bE\Big[X_{j\pat}\prod_{\pat'<\pat}(1-X_{j\pat'})\Big].
\end{align}
We will analyze the expectation in~\eqref{eq:arbalg} for an arbitrary $j$ and $\pat$.
We assume $\sss_\pat>0$, since otherwise $X_{j\pat}=0$ w.p.~1 and the expectation must be 0.

\textbf{Step 1: Using conditional independence.}
The terms inside the expectation in~\eqref{eq:arbalg} are generally not independent, since $X_{j\pat}=1$ (i.e.\ probing $j$ on attempt $\pat$) is only possible if $S_{\pat}=1$ (i.e.\ the online vertex is still available on attempt $\pat$), which is affected by whether $X_{j\pat'}=1$ for $\pat'<\pat$.
Nonetheless, we can decompose the expectation as follows:
\begin{align*}
\bE\Big[X_{j\pat}\prod_{\pat'=1}^{\pat-1}(1-X_{j\pat'})\Big]
&=\bE[X_{j\pat}]\bE\Big[\prod_{\pat'=1}^{\pat-1}(1-X_{j\pat'})\Big|X_{j\pat}=1\Big]
\\ &=\ssx_{j\pat}\prod_{\pat'=1}^{\pat-1}\bE[(1-X_{j\pat'})|X_{j\pat}=1]
\\ &=\ssx_{j\pat}\prod_{\pat'=1}^{\pat-1}\bE[(1-X_{j\pat'})|S_{\pat'+1}=1].
\end{align*}
The second equality holds because $\{X_{j\pat'}:\pat'<\pat\}$ were only correlated through their dependence on the availability of the online vertex, but conditional on $X_{j\pat}=1$ which reveals that the online vertex is available for all attempts $\pat'<\pat$, whether $j$ is probed on any attempt is determined by an independent coin flip.
The third equality holds because from the perspective of $X_{j\pat'}$, conditioning on $X_{j\pat}=1$ is equivalent to conditioning on $S_{\pat'+1}=1$.
Indeed, the only future information which distorts the likelihood of having probed $j$ on attempt $\pat'$ is whether the online vertex survived to attempt $\pat'+1$; any events after that have no interdependence with $X_{j\pat'}$.

For an arbitrary $\pat'$, noting that $S_{\pat'+1}=1$ implies $S_{\pat'}=1$, we can evaluate
\begin{align} \label{eq:temp}
\bE[(1-X_{j\pat'})|S_{\pat'+1}=1]
&=\frac{\Pr[X_{j\pat'}=0\cap S_{\pat'+1}=1|S_{\pat'}=1]}{\Pr[S_{\pat'+1}=1|S_{\pat'}=1]}.
\end{align}
Now, since $X_{j\pat'}=0$ and $X_{j\pat'}=1$ are mutually exclusive and collectively exhaustive events,
\begin{align*}
\Pr[S_{\pat'+1}=1|S_{\pat'}=1]
&=\Pr[X_{j\pat'}=0\cap S_{\pat'+1}=1|S_{\pat'}=1]+\Pr[X_{j\pat'}=1\cap S_{\pat'+1}=1|S_{\pat'}=1]
\\ &\le\Pr[X_{j\pat'}=0\cap S_{\pat'+1}=1|S_{\pat'}=1]+\ssx_{j\pat'}/\sss_{\pat'}
\end{align*}
where the inequality uses the fact that $\Pr[X_{j\pat'}=1|S_{\pat'}=1]=\ssx_{j\pat'}/\sss_{\pat'}$.
Substituting into~\eqref{eq:temp},
\begin{align*}
\bE[(1-X_{j\pat'})|S_{\pat'+1}=1]
&\ge1-\frac{\ssx_{j\pat'}/\sss_{\pat'}}{\Pr[S_{\pat'+1}=1|S_{\pat'}=1]}
\\&=1-\frac{\ssx_{j\pat'}}{\sss_{\pat'+1}}.
\end{align*}

In combination with the fact that expectations are non-negative, we can substitute back into expression~\eqref{eq:arbalg} to conclude that for any node $j$,
\begin{align} \label{eq:tempTwo}
\sum_{\pat=1}^n\bE\Big[X_{j\pat}\prod_{\pat'=1}^{\pat-1}(1-X_{j\pat'})\Big]
\ge\sum_{\pat:\sss_{\pat}>0}\ssx_{j\pat}\prod_{\pat'=1}^{\pat-1}\max\{1-\frac{\ssx_{j\pat'}}{\sss_{\pat'+1}},0\}.
\end{align}

\textbf{Step 2: Using the improved constraints.}
Now that in~\eqref{eq:tempTwo} we have bounded the algorithm's expected weight matched as an expression of variables from the LP's optimal solution, we can use the improved constraints~\eqref{eq:arblp-1pernode} to derive a bound relative to the LP's optimal objective value.

For this part of the proof, fix an arbitrary $j$ and we will omit the subscript $j$.
To ensure that we don't divide by 0 on the RHS of~\eqref{eq:tempTwo}, let $\opat$ denote the maximum index $\pat$ for which $\ssx_{\pat}>0$; note it must be the case that $\sss_{\opat}>0$.
We can then apply~\eqref{eq:arblp-1pernode} on the values of $\sss_{\pat'+1}$ to see that the RHS of~\eqref{eq:tempTwo} is lower-bounded by
\begin{align} \label{eqn:desiredResult}
\sum_{\pat=1}^{\opat}\ssx_{\pat}\prod_{\pat'=1}^{\pat-1}\max\{1-\frac{\ssx_{\pat'}}{\ssx_{\pat'+1}+\cdots+\ssx_{\opat}},0\}.
\end{align}
 
\begin{lemma} \label{lem:arbPatTechnical}
The algorithm's probability of probing any offline node in a non-simulated fashion is at least half the total from the LP optimal solution, when both sums are truncated at the index $\opat$ defined above.  That is,
$
\sum_{\pat=1}^{\opat}\ssx_{\pat}\prod_{\pat'=1}^{\pat-1}\max\{1-\frac{\ssx_{\pat'}}{\ssx_{\pat'+1}+\cdots+\ssx_{\opat}},0\}\ge\frac{1}{2}\sum_{\pat=1}^{\opat}\ssx_{\pat}.
$
\end{lemma}
The proof of \Cref{lem:arbPatTechnical} is presented next.
It directly completes the present proof from expression~\eqref{eqn:desiredResult}.
\Halmos\end{proof}

The analysis in Theorem~\ref{thm:arbpat} is tight.
Indeed, consider an example with $n=2$ offline vertices, where the first has $w_1=1$ and $p_1=\epsilon$ for some small $\epsilon>0$, while the second has $w_2=0$ and $p_2=1$.
The patience is deterministically 2, i.e. $q_1=q_2=1$.
The following defines an optimal LP solution: $s_1=1$, $x_{11}=\frac{1}{2(1-\epsilon)}$, $x_{21}=\frac{1-2\epsilon}{2(1-\epsilon)}$, $s_2=1/2$, $x_{12}=\frac{1-2\epsilon}{2(1-\epsilon)}$, $x_{22}=0$, which has objective value $\epsilon$ (all coming from the first vertex).
However, the algorithm can only probe the first vertex in a non-simulated fashion on attempt 1, since conditional on it surviving to attempt 2, the first vertex must have already been probed on attempt 1.
As a result, the algorithm probes the first vertex with a total probability of $\frac{1}{2(1-\epsilon)}$, matching expected weight $\frac{\epsilon}{2(1-\epsilon)}$ which approaches half of the LP's objective value as $\epsilon\to0$.

%%% Local Variables:
%%% mode: latex
%%% TeX-master: "follow_your_star_OR"
%%% End:

\begin{proof}{Proof of \Cref{lem:arbPatTechnical}.}
To simplify notation in the proof of~\eqref{eqn:desiredResult}, we let $\ssx_{>\pat'}$ denote $\ssx_{\pat'+1}+\cdots+\ssx_{\opat}$ for any $\pat'$; we also let $[\cdot]^+$ denote $\max\{\cdot,0\}$.
If $\opat=1$ then the desired result is trivial; if $\opat=2$ then~\eqref{eqn:desiredResult} equals $\ssx_1+\ssx_2[1-\ssx_1/\ssx_2]^+=\max\{\ssx_1,\ssx_2\}$ which also leads to the desired result that it is at least $\frac{1}{2}\sum_{\pat=1}^{\opat}\ssx_{\pat}$.
Hereafter, we assume $\opat\ge3$.  The following can then be derived:
\begin{align}
\sum_{\pat=1}^{\opat}\ssx_{\pat}\prod_{\pat'=1}^{\pat-1}\left[1-\frac{\ssx_{\pat'}}{\ssx_{>\pat'}}\right]^+
&=\sum_{\pat=1}^{\opat-2}\ssx_{\pat}\prod_{\pat'=1}^{\pat-1}\left[1-\frac{\ssx_{\pat'}}{\ssx_{>\pat'}}\right]^+
+\left(\ssx_{\opat-1}+\ssx_{\opat}\left[1-\frac{\ssx_{\opat-1}}{\ssx_{\opat}}\right]^+\right)\prod_{\pat'=1}^{\opat-2}\left[1-\frac{\ssx_{\pat'}}{\ssx_{>\pat'}}\right]^+ \notag
\\ &=\sum_{\pat=1}^{\opat-2}\ssx_{\pat}\prod_{\pat'=1}^{\pat-1}\left[1-\frac{\ssx_{\pat'}}{\ssx_{>\pat'}}\right]^+
+\max\{\ssx_{\opat-1},\ssx_{\opat}\}\frac{[\ssx_{>\opat-2}-\ssx_{\opat-2}]^+}{\ssx_{>\opat-2}}\prod_{\pat'=1}^{\opat-3}\left[1-\frac{\ssx_{\pat'}}{\ssx_{>\pat'}}\right]^+ \notag
\\ &\ge\sum_{\pat=1}^{\opat-2}\ssx_{\pat}\prod_{\pat'=1}^{\pat-1}\left[1-\frac{\ssx_{\pat'}}{\ssx_{>\pat'}}\right]^+
+\frac{\ssx_{>\opat-2}-\ssx_{\opat-2}}{2}\prod_{\pat'=1}^{\opat-3}\left[1-\frac{\ssx_{\pat'}}{\ssx_{>\pat'}}\right]^+ \label{eq:leftOff}
\end{align}
where the inequality applies the facts that $\frac{\max\{\ssx_{\opat-1},\ssx_{\opat}\}}{\ssx_{>\opat-2}}\ge1/2$ and $[\ssx_{>\opat-2}-\ssx_{\opat-2}]^+\ge\ssx_{>\opat-2}-\ssx_{\opat-2}$, after acknowledging that all terms in the product are initially non-negative.

Now, for any $k=\opat-2,\opat-1,\ldots,2$, the following can be derived:
\begin{align}
&\sum_{\pat=1}^k\ssx_{\pat}\prod_{\pat'=1}^{\pat-1}\left[1-\frac{\ssx_{\pat'}}{\ssx_{>\pat'}}\right]^+
+\frac{\ssx_{>k}-\ssx_k}{2}\prod_{\pat'=1}^{k-1}\left[1-\frac{\ssx_{\pat'}}{\ssx_{>\pat'}}\right]^+ \label{eq:indStart}
\\&=\sum_{\pat=1}^{k-1}\ssx_{\pat}\prod_{\pat'=1}^{\pat-1}\left[1-\frac{\ssx_{\pat'}}{\ssx_{>\pat'}}\right]^+
+\ssx_k\prod_{\pat'=1}^{k-1}\left[1-\frac{\ssx_{\pat'}}{\ssx_{>\pat'}}\right]^+ \notag
+\frac{\ssx_{>k}-\ssx_k}{2}\prod_{\pat'=1}^{k-1}\left[1-\frac{\ssx_{\pat'}}{\ssx_{>\pat'}}\right]^+
\\&=\sum_{\pat=1}^{k-1}\ssx_{\pat}\prod_{\pat'=1}^{\pat-1}\left[1-\frac{\ssx_{\pat'}}{\ssx_{>\pat'}}\right]^+
+\frac{\ssx_{>k-1}}{2}\frac{[\ssx_{>k-1}-\ssx_{k-1}]^+}{\ssx_{>k-1}}\prod_{\pat'=1}^{k-2}\left[1-\frac{\ssx_{\pat'}}{\ssx_{>\pat'}}\right]^+ \notag
\\&\ge\sum_{\pat=1}^{k-1}\ssx_{\pat}\prod_{\pat'=1}^{\pat-1}\left[1-\frac{\ssx_{\pat'}}{\ssx_{>\pat'}}\right]^+
+\frac{\ssx_{>k-1}-\ssx_{k-1}}{2}\prod_{\pat'=1}^{k-2}\left[1-\frac{\ssx_{\pat'}}{\ssx_{>\pat'}}\right]^+. \label{eq:indEnd}
\end{align}
When $k=\opat-2$, expression~\eqref{eq:indStart} equals~\eqref{eq:leftOff}.  We have inductively established that this is at least the value of expression~\eqref{eq:indEnd} when $k=2$, which is $\ssx_1+\frac{\ssx_{>1}-\ssx_1}{2}=\frac{1}{2}(\ssx_1+\cdots+\ssx_{\opat})$, completing the proof.
\Halmos\end{proof}

%%% Local Variables:
%%% mode: latex
%%% TeX-master: "follow_your_star_OR"
%%% End:

\begin{proof}{Proof of \Cref{thm:constantHazard}.}
\newcommand{\quit}{q}
\newcommand{\q}{\quit}
Let $\quit_{i} = p_{i} + (1-p_{i})r_{i}$ be the probability that probing terminates after attempting to match $i$ (which occurs either when the match is successful, or when the match fails and the patience is subsequently exhausted); the probability that another match can be attempted is then $1-\quit_{i}$.

%By an exchange argument, we show that probing in decreasing order of $w_{i}p_{i} / q_{i}$ provides an optimal ordering.

% To see this, first suppose wlog that $w_{1}p_{1} / q_{1} \ge w_{2}p_{2} . q_{2} \ge \dots \ge w_{m}p_{m}/q_{m}$.

% Next, let $\pi$ be an optimal permutation, so that probing in the order of
% $u_{\pi(1)}, u_{\pi(2)}, \dots, u_{\pi(m)}$ achieves the maximum possible
% expected reward. Let $\bE[\pi]$ denote this expected reward.
Suppose we have an optimal strategy which probes in the order $u_{1}, u_{2}, \dots, u_{m}$ and achieves the maximum expected reward. We denote this reward by $R^{*}$, and note that
\begin{equation}
\label{eq:chr-opt}%
R^{*} = \sum_{i=1}^{m} \left[ \prod_{j=1}^{i-1}(1-\quit_{j}) w_{i} p_{i}\right]
\end{equation}

If the permutation satisfies $w_{i}p_{i} / \quit_{i} \ge w_{i+1}p_{i+1} / \quit_{i+1}$ for all $i\in[m-1]$, then we are done.
% If $\pi$ is the identity permutation, then $\pi$ probes in the order
% $u_{1}, u_{2}, \dots, u_{m}$, i.e., in decreasing order of
% $w_{i}p_{i} / \quit_{i}$, and we are done.
Suppose to the contrary that there is some $k\in[m-1]$ for which
$w_{k}p_{k} / \quit_{k} < w_{k+1} p_{k+1} / \quit_{k+1}$. Consider the alternate strategy where we swap the order of probing $k$ and $k+1$. That is, we probe in the order $u_{1}, u_{2}, \dots, u_{k-1}, u_{k+1}, u_{k}, u_{k+2}, \dots, u_{m}$. Denote the expected reward of this new strategy by $R'$. We have
\begin{equation}
\label{eq:chr-optswap}%
R' = \sum_{i=1}^{k-1}\left[ \prod_{j=1}^{i-1}(1-\quit_{j}) w_{i}p_{i}\right] + \prod_{j=1}^{k-1}(1-\quit_{j}) \left( w_{k+1}p_{k+1} + (1-\quit_{k+1})w_{k}p_{k} \right) + \sum_{i=k+1}^{m} \left[ \prod_{j=1}^{i-1}(1-\quit_{j}) w_{i}p_{i}\right]
\end{equation}
Subtracting \eqref{eq:chr-optswap} from~\eqref{eq:chr-opt} yields
\begin{equation*}
R^{*} - R' = \prod_{j=1}^{k-1}(1-\quit_{j})\big( w_{k}p_{k} + (1-\quit_{k}) w_{k+1}p_{k+1} \big)
-
\prod_{j=1}^{k-1}(1-\quit_{j}) \big( w_{k+1}p_{k+1} - (1-\quit_{k+1}) w_{k}p_{k} \big)
\end{equation*}
which is equivalent to
\begin{align*}
R^{*} - R' &= \left[ \prod_{j=1}^{k-1}(1-\quit_{j}) \right] \bigg(w_{k}p_{k} + (1-q_{k})w_{k+1}p_{k+1} - w_{k+1}p_{k+1} - (1-q_{k+1})w_{k}p_{k} \bigg) \\
&= \left[ \prod_{j=1}^{k-1}(1-\quit_{j}) \right] \bigg( q_{k+1}w_{k}p_{k} - q_{k}w_{k+1}p_{k+1} \bigg) < 0
\end{align*}
where the inequality holds since, by assumption,
$w_{k}p_{k} / q_{k} < w_{k+1}p_{k+1} / q_{k+1}$, and thus
$w_{k}p_{k} q_{k+1} < w_{k+1}p_{k+1} q_{k}$.

However, this means that $R^{*} < R'$, contradicting the assumption that $R^{*}$ is optimal.
\Halmos
\end{proof}
%%% Local Variables:
%%% mode: latex
%%% TeX-master: "follow_your_star_OR"
%%% End:

\begin{proof}{Proof of \Cref{lem:itemarrivals-large}.}
  For a large item $i$, let $A_{i,t}$ be the event that $i$ arrives before or at
  time $t$. Let $B_{i,t}$ be the event that $i$ is offered before time $t$. Let
  $C_{i,t}$ be the event that we have successfully sold some item in
  $\LRG-\{i\}$ before time $t$.

  We note that for an item $i$ and time step $t$:
  \begin{equation}
    \label{eq:i-probedat-t}
    \Pr(i\text{ is probed at time }t) = \Pr(A_{i,t}\land \bar{B}_{i,t} \land \bar{C}_{i,t})\frac{\rho x_{i}}{\pat}
  \end{equation}
  Next, we observe that the probability that $i$ has arrived by time $t$ is given by
  \begin{equation}
    \label{eq:i-arrive-t}
    \Pr(A_{i,t}) = 1 - (1-q_{i})^{t}.
  \end{equation}
  Next, we consider $\Pr(\bar{B}_{i,t} \cap \bar{C}_{i,t})$, which is given by:
  \begin{equation}
    \label{eq:BitCit}
    \Pr(\bar{B}_{i,t}\cap \bar{C}_{i,t}) = \prod_{t'=1}^{t-1} \Pr(\bar{B}'_{i,t'}\cap \bar{C}'_{i,t'})
  \end{equation}
  where $B'_{i,t'}$ is the event that $i$ is offered at time $t'$ and
  $C'_{i,t'}$ is the event that some item in $\LRG-\{i\}$ is offered at time
  $t'$. Next, we note that
  $\Pr(\bar{B}'_{i,t'}\cap \bar{C}'_{i,t'}) = 1 - \Pr(B'_{i,t'} \lor C'_{i,t'}) \ge 1 - \Pr(B'_{i,t'}) - \Pr(C'_{i,t'})$.
  Then,
  \[
    \Pr(B'_{i,t'}) \le \frac{\rho x_{i}}{\pat} \le \frac{\rho}{\pat}
  \]
  and
  \[
    \Pr(C'_{i,t}) \le \sum_{j} \frac{\rho x_{j}}{\pat}p_{j}
    = \frac{\rho}{\pat}\sum_{j} x_{j}p_{j}  \le \frac{\rho}{\pat}
  \]
  where in the final inequality, we use constraint~\eqref{eq:lp-1match}. These
  together give
  \begin{equation*}
    \Pr(\bar{B}'_{i,t'} \cap \bar{C}'_{i,t'}) \ge 1 - \frac{2\rho}{\pat}
  \end{equation*}
  which combined with~\eqref{eq:BitCit} gives~\eqref{eq:barbbarc1minus2rho}.
  \begin{equation}
    \label{eq:barbbarc1minus2rho}
    \Pr(\bar{B}_{i,t'} \cap \bar{C}_{i,t'}) \ge \left( 1 - \frac{2\rho}{\pat} \right)^{t-1} \ge \left(1 - \frac{2\rho}{\pat}\right)^{\pat}
  \end{equation}
  Combining~\eqref{eq:i-arrive-t} and~\eqref{eq:barbbarc1minus2rho}
  with~\eqref{eq:i-probedat-t} gives us~\eqref{eq:iprobedattbound}.

  \begin{equation}
    \label{eq:iprobedattbound}
    \Pr(i\text{ is probed at time }t)
    \ge \left(1 - (1-q_{i})^{t} \right) \left(1-\frac{2\rho}{\pat}\right)^{\pat}\frac{\rho x_{i}}{\pat}
  \end{equation}

We now use~\eqref{eq:iprobedattbound} to derive a bound on $\Pr(i\text{ gets
    probed})$.
  \begin{align*}
    \Pr(i\text{ gets probed})
    &\ge \frac{\rho x_{i}}{\pat}\left(1 - \frac{2\rho}{\pat}\right)^{\pat}\sum_{t=1}^{\pat} \left(1 - (1-q_{i})^{t}\right) \\
    &= \frac{\rho x_{i}}{\pat}\left(1 - \frac{2\rho}{\pat}\right)^{\pat} \left(\pat - \sum_{t=1}^{\pat}(1-q_{i})^{t}\right) \\
    &\ge \frac{\rho x_{i}}{\pat}\left(1 - 2\rho\right) \left(\pat - \sum_{t=1}^{\pat}(1-q_{i})^{t}\right)
  \end{align*}
  Then, we note that
  \[
    \sum_{t=1}^{\pat}(1-q_{i})^{t} = (1-q_{i}) \frac{1-(1-q_{i})^{\pat}}{q_{i}}
  \]
  For a large item $i$, we have
  \[
    (1-q_{i}) \frac{1-(1-q_{i})^{\pat}}{q_{i}} \le \frac{1-c/\pat}{c/\pat}\left(1 - (1-c/\pat)^{\pat}\right)
    = \frac{\pat}{c} \left(1-\frac{c}{\pat}\right)(1 - (1-c/\pat)^{\pat})
  \]

  where the inequality follows from the fact that the first expression is
  decreasing in $q_{i}$ (and, for large items, $q_{i} \ge c/\pat$).

  This finally gives us a bound for large $i$:

  \begin{equation}\label{eq:large-iprobedbound1}
    \Pr(i\text{ gets probed})
    \ge \rho (1-2\rho)\left(1 - \frac{1}{c}\left(1-c/\pat\right)\left(1-(1-c/\pat)^{\pat}\right)\right) x_{i}
  \end{equation}
  We argue that the expression on the RHS is decreasing in $\pat$. To see this,
  let
  \[
    f_{c}(\pat) = \left(1 - \frac{c}{\pat}\right)\left(1-\left(1-\frac{c}{\pat}\right)^{\pat}\right)
  \]
  so that we can rewrite equation~\eqref{eq:large-iprobedbound1} as
  \[
    \Pr(i\text{ gets
      probed}) \ge \rho(1-2\rho)\left(1 - \frac{1}{c}f_{c}(\pat)\right)x_{i}
  \]
  To show that the RHS of~\eqref{eq:large-iprobedbound1} is decreasing in
  $\pat$, it suffices to show that $f_{c}(\pat)$ is increasing. Taking the
  derivative, we get:
\begin{equation}
  \label{eq:large-fcTderivative}
  \begin{split}
    \frac{df}{d\pat}
    &=
    \left(\frac{c}{\pat^{2}}\right) \left(1 - \left(1-\frac{c}{\pat}\right)^{\pat}\right) +
    \left(1 - \frac{c}{\pat}\right) \left(1-\frac{c}{\pat}\right)^{\pat}\ln\left(1 - \frac{c}{\pat}\right)\frac{c}{\pat^{2}} \\
    &= \frac{c}{\pat^{2}}\left( 1 - \left(1-\frac{c}{\pat}\right)^{\pat} + \left(1-\frac{c}{\pat}\right)^{\pat+1}\ln\left(1-\frac{c}{\pat}\right) \right) \\
    &= \frac{c}{\pat^{2}}\left( 1 - \left(1-\frac{c}{\pat}\right)^{\pat} \left( 1 - \left(1-\frac{c}{\pat}\right)\ln\left(1-\frac{c}{\pat}\right) \right) \right) \\
  \end{split}
\end{equation}
We wish to show that~\eqref{eq:large-fcTderivative} is nonnegative.
First, we observe that
\begin{equation} \label{eq:large-fcT-deriv-ln-bound}
  \begin{split}
    1 - \left(1-\frac{c}{\pat}\right)\ln\left(1-\frac{c}{\pat}\right)
    &\le 1 - \ln\left(1 - \frac{c}{\pat}\right) \\
    &< 1 + \frac{\frac{c}{\pat}}{1 - \frac{c}{\pat}}
    = \frac{1}{1-\frac{c}{\pat}}
  \end{split}
\end{equation}
where the second inequality follows from the fact that $\ln(1+x) > x / (1+x)$
when $x>-1$ (and taking $x=-c/\pat$). Combining~\eqref{eq:large-fcT-deriv-ln-bound}
with~\eqref{eq:large-fcTderivative}, we get:
\begin{equation*}
  \frac{df}{d\pat}
  \ge
  \frac{c}{\pat^{2}}\left( 1 - \left(1-\frac{c}{\pat}\right)^{\pat} \left( \frac{1}{1-\frac{c}{\pat}} \right) \right)
  = \frac{c}{\pat^{2}}\left( 1 - \left(1-\frac{c}{\pat}\right)^{\pat-1} \right)
  > 0
\end{equation*}
Therefore, $f_{c}(\pat)$ is increasing in $\pat$, and the RHS
of~\eqref{eq:large-iprobedbound1} is decreasing in $\pat$. As a result, the
expression is minimized when $\pat\to\infty$, so we can simplify our bound further.
  \begin{equation} \label{eq:largeiprobedbound}%
    \Pr(i\text{ gets probed})
    \ge \rho (1-2\rho)\left(1 - \frac{1}{c}\left(1 - e^{-c}\right)\right)x_{i}
  \end{equation}

  Thus, the total expected weight produced by $\pi_{\textsc{Large}}$ is at least
  \begin{equation}
    \label{eq:LRGbound}
    \begin{aligned}
      \sum_{i\in\LRG} w_{i}p_{i}\Pr(i\text{ gets probed})
      &\ge \sum_{i\in \LRG} w_{i} p_{i} \left(1-\frac{1}{c}\left(1 - e^{-c}\right)\right) \rho(1-2\rho) \\
      &\ge \left(1 - \frac{1}{c}\left(1 - e^{-c}\right)\right)\rho(1-2\rho) (1-\varphi)\mathrm{LP}
    \end{aligned}
  \end{equation}
  where the first inequality comes from~\eqref{eq:largeiprobedbound} and the
  second comes from the assumption.
  \Halmos
\end{proof}

%%% Local Variables:
%%% mode: latex
%%% TeX-master: "follow_your_star_OR"
%%% End:

\begin{proof}{Proof of \Cref{lem:itemarrivals-small}.}
  Since $\sum_{i\in \LRG} w_{i}x_{i}p_{i} < (1-\varphi) \mathrm{LP}$, we must
  have that $\sum_{i\in \SML} w_{i}x_{i}p_{i} \ge \varphi\mathrm{LP}$.

  As before, we consider the probability that an item $i$ is probed at time step
  $t$. Let $A_{i,t}$ be the event that $i$ arrives at time $t$, $B_{t}$ be the
  event that an item is successfully sold before time $t$, and $C_{i,t}$ be the
  event that some other small item $j\ne i$ arrives at time $t$ and wants to be
  probed.
  \begin{align*}
    \Pr(i\text{ probed at time } t)
    &= \Pr(A_{i,t}\cap \bar{B}_{t} \cap \bar{C}_{i,t}) \frac{x_{i}}{1 - (1-q_{i})^{\pat}} \\
    &= \Pr(A_{i,t}) \Pr(\bar{B}_{t} \cap \bar{C}_{i,t} \mid A_{i,t}) \frac{x_{i}}{1 - (1-q_{i})^{\pat}} \\
    &= \Pr(A_{i,t}) \Pr(\bar{B}_{t} \mid A_{i,t}) \Pr(\bar{C}_{i,t} \mid A_{i,t}\cap \bar{B}_{t}) \frac{x_{i}}{1 - (1-q_{i})^{\pat}} \\
    &= \Pr(A_{i,t}) \left(1 - \Pr(B_{t} \mid A_{i,t})\right) \Pr(\bar{C}_{i,t} \mid A_{i,t} \cap \bar{B}_{t}) \frac{x_{i}}{1 - (1-q_{i})^{\pat}}
  \end{align*}

  The probability of $i$ arriving precisely at time $t$ is
  \[
    \Pr(A_{i,t}) = q_{i}(1 - q_{i})^{t-1}.
  \]

  For the probability that an item is successfully sold before time $t$, since
  we condition on $i$ arriving at time $t$ we only need to consider the
  probability of successfully selling some item $j\ne i$:
  \begin{align*}
    \Pr(B_{t} \mid A_{i,t})
    &\le \sum_{j\in \SML-\{i\}} \left(\sum_{t'=1}^{t-1} q_{j}(1-q_{j})^{t'-1} \frac{x_{j}}{1 - (1-q_{j})^{\pat}}p_{j} \right) \\
    &= \sum_{j\in \SML-\{i\}} (1 - (1 - q_{j})^{t-1})\frac{x_{j}}{1 - (1-q_{j})^{\pat}}p_{j} \\
    &= \sum_{j\ne \SML-\{i\}} \frac{1 - (1 - q_{j})^{t-1}}{1 - (1-q_{j})^{\pat}} x_{j}p_{j}
  \end{align*}
  Next, we note that the expression
  \[
    \frac{1 - (1-q_{j})^{t-1}}{1 - (1-q_{j})^{\pat}}
  \]
  is increasing in $q_{j}$ and so we have
  \[
    \Pr(B_{t} \mid A_{i,t})
    \le \sum_{j\ne \SML-\{i\}} \frac{1 - (1-c/\pat)^{t-1}}{1-(1-c/\pat)^{\pat}} x_{j}p_{j}
  \]
  since $q_{j} \le c/\pat$ for small $j$. Using LP constraint~\eqref{eq:lp-1match}
  then gives~\eqref{eq:smallBtbound}.
  \begin{equation}
    \label{eq:smallBtbound}
    \Pr(B_{t} \mid A_{i,t})
    \le \frac{1 - (1-c/\pat)^{t-1}}{1-(1-c/\pat)^{\pat}}
  \end{equation}

  For the probability that another item $j\ne i$ arrives and wants to be probed, we have:
  \begin{equation*}
    \Pr(\bar{C}_{i,t})
    \ge \prod_{j\in \SML-\{i\}} \left(1 - (1-q_{j})^{t-1}q_{j} \frac{x_{j}}{1 - (1-q_{j})^{\pat}}\right)
    \ge \prod_{j\in \SML-\{i\}} \left(1 - q_{j} \frac{x_{j}}{1 - (1-q_{j})^{\pat}}\right)
  \end{equation*}
  We wish to lower bound this product. Let
  $z_{t,j} := q_{j} \frac{x_{j}}{1 - (1-q_{j})^{\pat}}$. Notice that since
  $q_{j} \le c/\pat$ (because $j$ is small) and $x_{j}/(1 - (1-q_{j})^{\pat}) \le 1$
  (by constraint~\eqref{eq:lp-xibound}), we have that $z_{j} \le c/\pat$. Further:
  \[
    \sum_{j\in \SML-\{i\}} z_{j}
    = \sum_{j\in \SML-\{i\}} q_{j} \frac{x_{j}}{1 - (1-q_{j})^{\pat}}
    \le \sum_{j\in \SML-\{i\}} \frac{c}{\pat} \frac{x_{j}}{1 - (1-\frac{c}{\pat})^{\pat}}
    \le \frac{c}{1 - (1-\frac{c}{\pat})^{\pat}}
  \]
  A lower bound is therefore given by the solution to the following minimization
  problem:
  \begin{align*}
    \min &\prod_{j\in \SML-\{i\}} (1 - z_{j}) \\
    \text{s.t. } & z_{j} \le c/\pat, \quad \forall j\in \SML-\{i\} \\
    & \sum_{j\in \SML-\{i\}} z_{j} \le \frac{c}{1 - (1-\frac{c}{\pat})^{\pat}}
  \end{align*}

  Lemma~\ref{lem:small-minprob-bound} (the precise statement and proof of which
  is presented after the present proof) gives us a lower bound on the solution
  to this minimization problem. Using $b=c/\pat$ and
  $S=c / (1 - (1-c/\pat)^{\pat})$, Lemma~\ref{lem:small-minprob-bound} tells us
  that
  \begin{equation}
    \label{eq:small-barCitbound}
    \Pr(\bar{C}_{i,t} \mid A_{i,t} \cap \bar{B}_{t}) \ge
    \left(1-\frac{c}{\pat}\right)^{\frac{\pat}{1 - (1-c/\pat)^{\pat}}}
  \end{equation}

  We can now put everything together to get
  \[
    \Pr(i\text{ probed at time }t)
    \ge q_{i}(1-q_{i})^{t-1}\left( 1
      - \frac{1 - (1-\frac{c}{\pat})^{t-1}}{1 - (1-\frac{c}{\pat})^{\pat}}
    \right)
    \left(1-\frac{c}{\pat}\right)^{\frac{\pat}{1 - (1-c/\pat)^{\pat}}}
    \frac{x_{i}}{1 - (1-q_{i})^{\pat}}.
  \]
  By summing over all time steps $t$, we can get a bound on the probability that a small item $i$ gets probed:
  \begin{align*}
    \Pr(i\text{ gets probed})
    &\ge \sum_{t=1}^{\pat} q_{i}(1-q_{i})^{t-1}\frac{x_{i}}{1 - (1-q_{i})^{\pat}}
      \left( 1
      - \frac{1 - (1-\frac{c}{\pat})^{t-1}}{1 - (1-\frac{c}{\pat})^{\pat}}
      \right)
      \left( 1 - \frac{c}{\pat} \right)^{\frac{\pat}{1 - (1-c/\pat)^{\pat}}} \\
    &= \sum_{t=1}^{\pat} q_{i}(1-q_{i})^{t-1}\frac{x_{i}}{1 - (1-q_{i})^{\pat}}
      \left( \frac{1-(1-\frac{c}{\pat})^{\pat} - (1 - (1-\frac{c}{\pat})^{t-1})}{1 - (1 - \frac{c}{\pat})^{\pat}} \right)
      \left( 1 - \frac{c}{\pat} \right)^{\frac{\pat}{1 - (1-c/\pat)^{\pat}}} \\
    &\ge \left(\frac{1}{\pat} \sum_{t=1}^{\pat} \frac{1-(1-\frac{c}{\pat})^{\pat} - (1 - (1-\frac{c}{\pat})^{t-1})}{1 - (1 - \frac{c}{\pat})^{\pat}} \right)
      \left( 1 - \frac{c}{\pat} \right)^{\frac{\pat}{1 - (1-c/\pat)^{\pat}}}
      \sum_{t=1}^{\pat}q_{i}(1-q_{i})^{t-1}\frac{x_{i}}{1 - (1-q_{i})^{\pat}} \\
    &\ge \left(\frac{1}{\pat} \sum_{t=1}^{\pat} \frac{1-(1-\frac{c}{\pat})^{\pat} - (1 - (1-\frac{c}{\pat})^{t-1})}{1 - (1 - \frac{c}{\pat})^{\pat}} \right)
      \left( 1 - \frac{c}{\pat} \right)^{\frac{\pat}{1 - (1-c/\pat)^{\pat}}}
      x_{i} \\
    &= \frac{1}{\pat} \frac{1}{1 - (1-\frac{c}{\pat})^{\pat}}
      \left( 1 - \frac{c}{\pat} \right)^{\frac{\pat}{1 - (1-c/\pat)^{\pat}}}
      \sum_{t=1}^{\pat} \left( 1- \left(1-\frac{c}{\pat}\right)^{\pat} - \left(1 - \left( 1-\frac{c}{\pat} \right)^{t-1} \right) \right)
      x_{i}
  \end{align*}
%   \wnote{I think the third line should be formally justified using the
    % rearrangement inequality?} \nnote{I had thought about this and was convinced
    % it was valid, but I'm not sure it follows (at least not trivially) from the
    % % rearrangement inequality? Or else, I am not seeing why it's immediate from
    % that. What was the justification you had in mind, Will?}
    Now, note that
  \begin{align*}
    &\sum_{t=1}^{\pat} \left( 1-(1-\frac{c}{\pat})^{\pat} - (1 - (1-\frac{c}{\pat})^{t-1}) \right) \\
    &= \sum_{t=1}^{\pat}\left( -(1-\frac{c}{\pat})^{\pat} + (1 - \frac{c}{\pat})^{t-1}  \right) \\
    &= -\pat(1 - \frac{c}{\pat})^{\pat} + \sum_{t=1}^{\pat}(1-\frac{c}{\pat})^{t-1} \\
    &= -\pat(1 - \frac{c}{\pat})^{\pat} + \frac{\pat}{c}(1 - (1-\frac{c}{\pat})^{\pat})
  \end{align*}
  which gives
  \[
    \Pr(i\text{ gets probed})
    \ge \left( 1 - \frac{c}{\pat} \right)^{\frac{\pat}{1 - (1-c/\pat)^{\pat}}}
    \left( \frac{1}{c} - \frac{(1-c/\pat)^{\pat}}{1 - (1-c/\pat)^{\pat}} \right) x_{i}
    \ge \left( 1 - \frac{c}{\pat} \right)^{\frac{\pat}{1 - (1-c/\pat)^{\pat}}}
    \left(\frac{1}{c} - \frac{e^{-c}}{1 - e^{-c}} \right) x_{i}
  \]
  since $(1-c/\pat)^{\pat}$ is maximized when $\pat\to \infty$ so that $(1-c/\pat)^{\pat} \le e^{-c}$. Further, we note that
  \[
    \left(1 - \frac{c}{\pat}\right)^{\frac{\pat}{1 - (1-c/\pat)^{\pat}}}
  \]
  is increasing in $\pat\ge 2$ when $c\in(0,1)$, and so we have
  \[
    \left(1 - \frac{c}{\pat}\right)^{\frac{\pat}{1 - (1-c/\pat)^{\pat}}} \ge
    \left(1 - \frac{c}{2}\right)^{\frac{2}{1 - (1-c/2)^{2}}}
  \]
  which allows is to get our final bound:
  \begin{align*}
    \sum_{i\in\SML} w_{i} p_{i} \Pr(i\text{ gets probed})
    &\ge \sum_{i\in\SML} w_{i} p_{i} x_{i} \left(\frac{1}{c} - \frac{e^{-c}}{1 - e^{-c}}\right)
    \left(1 - \frac{c}{2}\right)^{\frac{2}{1 - (1-c/2)^{2}}} \\
    &\ge \varphi
    \left(1 - \frac{c}{2}\right)^{\frac{2}{1 - (1-c/2)^{2}}}
    \left(\frac{1}{c} - \frac{e^{-c}}{1 - e^{-c}} \right)\mathrm{LP}
  \end{align*}
  as desired.
  \Halmos
\end{proof}

\begin{lemma}
  \label{lem:small-minprob-bound}
  Consider the following optimization problem over variables $z_{j}$ for
  $j\in J$, with fixed positive constants $b,S$:
  \begin{align}
    F = \min &\prod_{j\in J} (1 - z_{j}) \label{eq:small-minprob-objective} \\
    \textup{s.t. } & z_{j} \le b, \quad \forall j\in J \\
         & \sum_{j\in J} z_{j} \le S
  \end{align}
  Then, the solution satisfies
  \[
    F \ge (1 - b)^{S/b}
  \]
\end{lemma}
\begin{proof}{Proof of \Cref{lem:small-minprob-bound}.}
  First, we claim that in the optimal solution to the above program, at most one
  $j$ has $z_{j} \in (0,b)$; the rest must all be equal to either $0$ or $c/\pat$.
  To see this, note that if we have two items $j,k$ for which both
  $z_{j}\in(0,b)$ and $z_{k}\in(0,b)$, then the objective value is
  $Z\cdot (1-z_{j})(1-z_{k})$, where $Z = \prod_{l\in \SML-\{i,j,k\}}(1-z_{l})$.
  Assume wlog that $z_{j} \ge z_{k}$. Then, we consider the alternate assignment
  $z_{j}' = z_{j}+\epsilon$, $z_{k}' = z_{k} - \epsilon$, where
  $\epsilon = \min\left\{ b - z_{j}, z_{k}\right\}$. Clearly, since
  $0< z_{k} \le z_{j} < b$, we have $0 < \epsilon < b$. We can then see that
  \begin{equation*}
    \begin{aligned}
      \left(1-z_{j}'\right) \left(1-z_{k}'\right)
      &= \left(1 - z_{j} - \epsilon\right) \left(1 - z_{k}' + \epsilon\right) \\
      &= 1 - z_{j} - z_{k} + z_{j}z_{k} + \epsilon z_{k} - \epsilon z_{j} + \epsilon^{2} \\
      &= (1-z_{j})(1-z_{k}) - \epsilon (z_{j} - z_{k} + \epsilon) < (1-z_{j})(1-z_{k})
    \end{aligned}
  \end{equation*}
  since $\epsilon>0$ and $z_{j} \ge z_{k}$. This alternate assignment would
  therefore have an objective value of
  $Z (1-z_{j}')(1-z_{k}') < Z(1-z_{j})(1-z_{k})$.

  Further, note from the definition of $\epsilon$ that we have $z'_{j} \le b$
  and $z'_{k}\ge 0$, with at least one of these being equality. Certainly, we
  also have $z_{j}>0$ and $z_{k}< b$, and it is clear that
  $z_{j}' + z_{k}' = z_{j} + z_{k}$; thus, the new assignment is still feasible.
  This means that if we have two items $j,k$ with $0 < z_{k} \le z_{j} < b$,
  then this is not an optimal assignment since we can decrease the objective
  value by either increasing $z_{j}$ to $b$ or decreasing $z_{k}$ to $0$.

  We can therefore see that the objective is minimized when $z_{j} = b$ for as
  many items as possible. Let $\nu = \lfloor S / b \rfloor$, and
  $\eta = S - b \nu$. The optimal solution sets $z_{j} = b$ for $\nu$ items,
  $z_{j} = \eta$ for a single item, and $z_{j}=0$ for all other items. This
  gives an objective value of
  \[
    F = {\left(1 - b\right)}^{\nu}(1-\eta)
  \]
  and we wish to show that this is lower bounded by $(1-b)^{S/b}$.
  Noting that $S/b = \nu + \eta b$, we wish to show that
  \begin{equation*}
    {\left( 1 - b \right)}^{\nu}(1-\eta) \ge {\left( 1 - b \right)}^{\nu + \eta / b}
  \end{equation*}
  Simplifying this expression, we find it is equivalent to
  \begin{equation}
    \label{eq:etacoverTineq}
    {\left( 1 - \eta \right)}^{1/\eta} \ge {\left(1 - b\right)}^{1/b}
  \end{equation}

  The Taylor series expansion of the natural log gives that
  \[
    \frac{1}{x} \ln(1-x) = -\frac{1}{x}\sum_{r=1}^{\infty} \frac{x^{r}}{r!} = -\sum_{r=1}^{\infty} \frac{x^{r-1}}{r!}
  \]
  From this, it is clear that since $\eta \le b$, we must have
  \[
    \sum_{r=1}^{\infty}\frac{\eta^{r-1}}{r!} \le \sum_{r=1}^{\infty} \frac{b^{r-1}}{r!}
  \]
  and therefore
  $\frac{1}{\eta}\ln(1-\eta) \ge \frac{1}{b}\ln\left(1-\frac{1}{b}\right)$,
  which implies~\eqref{eq:etacoverTineq} as desired.
  \Halmos
\end{proof}

%%% Local Variables:
%%% mode: latex
%%% TeX-master: "follow_your_star_OR"
%%% End:

%%% Local Variables:
%%% mode: latex
%%% TeX-master: "follow_your_star_OR"
%%% End:

\begin{proof}{Proof of \Cref{thm:itemarrivals}}
  We use $\rho=1/4$ and $\varphi=0.37$. For this $\rho,\varphi$, strategy
  $\pi_{\textsc{Large}}$ achieves a ratio of
  \[
    (1-0.37)\left(\frac{1}{8}\right)\left(1 - \frac{1}{c}\left(1 - e^{-c}\right)\right)
  \]
  and strategy
  $\pi_{\textsc{Small}}$ achieves a ratio of
  \[
    0.37
    \left(1 - \frac{c}{2}\right)^{\frac{2}{1 - (1-c/2)^{2}}}
    \left(\frac{1}{c} - \frac{e^{-c}}{1 - e^{-c}} \right)
  \]
  The equation
  \[
    (1-0.37)\left(\frac{1}{8}\right)\left(1 - \frac{1}{c}\left(1 - e^{-c}\right)\right)
    =
    0.37
    \left(1 - \frac{c}{2}\right)^{\frac{2}{1 - (1-c/2)^{2}}}
    \left(\frac{1}{c} - \frac{e^{-c}}{1 - e^{-c}} \right)
  \]
  has a solution at $c\approx 0.9249$, where both $\pi_{\textsc{Large}}$ and
  $\pi_{\textsc{Small}}$ achieve an approximation ratio of approximately $0.0273714 > 0.027$.
  \Halmos
\end{proof}

%%% Local Variables:
%%% mode: latex
%%% TeX-master: "follow_your_star_OR"
%%% End:

\begin{proof}{Proof of \Cref{thm:mehtalpstochgap}.}
Consider the complete bipartite graph $\bpgraph$ with $|U|=|V|=n$ and $p_{uv}=1/n$ for all $(u,v)\in U\times V$. We state the following result that is implied by \citet{karp1981maximum}.
\begin{lemma}[Theorem~14 of~\citet{bb95randgraph}; implied by \citet{karp1981maximum}]  \label{lem:randgraph}
Let $G$ be a random bipartite graph with both partitions of size $n$ and where each edge exists independently with probability $p=1/n$. Let $\gamma$ be the solution to the equation $\gamma = e^{-\gamma}$. Then, the largest independent set of $G$ has size $n(2\gamma+\gamma^{2}))[1+o(1)]$ with probability $1-o(1)$.
\end{lemma}
\Cref{lem:randgraph} implies that, almost surely as $n\to\infty$, there exists an independent set of size $n(2\gamma+\gamma^2)$, and hence there exists a vertex cover of size $2n-n(2\gamma+\gamma^2)\approx 0.544n$ (found by taking all the vertices not in the independent set).
Therefore almost surely no matching can have size greater than $0.544n$.
It follows that no online or offline algorithm can achieve an expected matching size greater than
\[
\frac{(1-o(1))0.544n+o(1)n}{n}=0.544+o(1),
\]
completing the proof that the stochasticity gap is at least 0.544. 
\Halmos\end{proof}

\begin{proof}{Proof of \Cref{thm:halfub}.}
We complete the proof of \Cref{thm:halfub} by showing that an online algorithm cannot earn more than $k+o(k)$.

An adversary in the online setting may expose all
vertices of $V_{n}$ before any of the vertices of $V_{0}$ to the
greedy algorithm. \textup{SimpleGreedy} choosing arbitrarily may in the
worst case choose to probe edges of $E_{0}$ first, preventing some
vertices of $V_{0}$ from being matched later. We consider the
case where SimpleGreedy chooses an edge $(u,v_{i})\in E_{0}$ for each online vertex
$v_{i}\in V_{n}$ if any $u \in U_0$ is available at $v_i$'s arrival. We calculate the
expected size of the matching produced by this strategy.

Let $M$ be a random variable corresponding to the size of the final
matching, and let $M_{0}$ and $M_{n}$ be random variables
corresponding to the number of matched vertices in $V_{0}$ and
$V_{n}$, respectively. Then, the expected size of the matching is
\[
\bE[M] = \bE[M_{n}] + \bE[M_{0}]
= k + \bE[M_{0}].
\]

We now consider $\bE[M_{0}]$. If $l<k$ vertices of $V_{n}$
are matched successfully, then when the vertices of $V_{0}$ arrive
online, there will only be $k-l$ vertices of $U_{0}$ remaining to be
matched. Since $|V_{0}|=kn^{2}$, greedy will almost surely match all
of them successfully. Thus, we get (as $n\to\infty$)
\begin{equation*}
\begin{split}
\bE[M_{0}] &= \sum_{l=0}^{k-1}
\binom{n}{l}(1-k/n)^{n-l}(k/n)^{l}(k-l)
\sim \sum_{l=0}^{k-1}(k-l)\frac{n^{l}}{l!}e^{-k(n-l)/n}\frac{k^{l}}{n^{l}} \\
&\sim \sum_{l=0}^{k-1}(k-l)\frac{e^{-k}k^{l}}{l!}  =
k\sum_{l=0}^{k-1}\frac{e^{-k}k^{l}}{l!} - \sum_{l=1}^{k-1}\frac{e^{-k}k^{l}}{(l-1)!} \\
&= k\left[ \sum_{l=0}^{k-1}\frac{e^{-k}k^{l}}{l!} -
\sum_{i=0}^{k-2}\frac{e^{-k}k^{l}}{l!}\right] =
k\cdot\frac{e^{-k}k^{k-1}}{(k-1)!} =
k\cdot\frac{e^{-k}k^{k}}{k!}.
\end{split}
\end{equation*}
Finally, we observe that for large $k$,
$e^{-k}k^{k}/k! \sim (2\pi k)^{-1/2}$, due to Stirling's formula.
Therefore, $\bE[M]=k+o(k)$.
%Thus, we get a competitive ratio of
%\begin{equation}
%\label{eq:cr}
%\frac{\bE[M]}{\mathrm{OPT}} =
%\lim_{k\to\infty}\frac{k + \sqrt{k/(2\pi)}}{2k} = \frac{1}{2}
%\end{equation}
\Halmos\end{proof}

\begin{proof}{Proof of \Cref{thm:unknownpatience-badcr}.}
Fix a positive integer $k$ and let $m = \omega(k)$ be large.
% ~\nnote{I added this $m=\omega(k)$ bit so that $k/m\to 0$ asymptotically, which we use at the end of the proof.}
An algorithm which knows the patience in advance may adopt the
following strategy.
% If the patience of $v$ is $1$, probe $(u_{0},v)$
% to earn an expected reward of $1$. Otherwise, if the patience is
% $m^{2i}$ for some $i\in\set{1,2,\dots,k-1}$, probe all
% $m^{2i}$ edges $(u,v)$ for which $u\in U_{i}$ (i.e.\ each edge
% incident to a vertex of $U_{i}$).
For any $i=0,\ldots,k$, if the patience of $v$ is $m^{2i}$, then repeatedly probe edges for offline vertices of type $i$, which have weight $m^i$ and probability $m^{-2i}$, until a success or until patience expires.
If any of these edges is
successfully matched, a reward of $m^{i}$ is achieved, so the
expected reward in this case is
$[1 - (1 - m^{-2i})^{m^{2i}}] m^{i}\ge(1-1/e)m^i$. Thus, the expected
reward of this strategy is at least
\begin{align*}
&\sum_{i=0}^{k-1}(m^{-i}-m^{-i-1})(1-1/e)m^i+m^{-k}(1-1/e)m^k
\\ &= \sum_{i=0}^{k-1} (1-1/m)(1-1/e) + (1-1/e)
%\\ &= k(1-1/m)(1-1/e) + (1-1/e)
\\ &=(k+1)(1-1/e)-k/m(1-1/e)
\\ &=(1-1/e)(k+1-k/m).
\end{align*}
% which is $\Theta(k)$ as $m\to\infty$. \nnote{I worked this out, and I think it's
%   basically $\sim k(1 - 1/\euler) - 1/m (k + 1 - 1/\euler) \sim k(1-1/\euler)$, but I think it's fine to just leave it at this? (I think it's not too hard to see at a glance that the summation approaches $k\cdot 1/\euler$, and the other subtracted terms go to zero, as $m\to\infty$).}
%As $m\to\infty$, this approaches
%$1/m^{k} + (1-1/m)k(1 - 1/e) \sim k(1-1/e) = \Theta(k)$.

We now show that an online algorithm cannot achieve an expected reward greater
than $1$.
%% NEW PROOF BELOW
% \textcolor{Green}{To do this, we consider the optimal expected reward achievable
%   by an online algorithm on a slightly modified version of our graph. In the
%   modified graph, we simply add $m^{2k}$ dummy vertices $U_{\varnothing}$, and
%   for each dummy vertex $u\in U_{\varnothing}$, the edge $e = (u, v)$ has
%   probability $p_{e} = 0$. Probing these edges achieves no reward, and
%   contributes nothing to the expected reward of the online algorithm. Thus, no
%   algorithm can ever benefit from probing these edges, and their presence does
%   not affect the optimal strategy or its expected reward. However, their
%   presence will be useful in the inductive argument below.}
We split the online probing strategy into \emph{stages}, where in stage $i$ it is known
that the patience is at least $m^{2i}$, for all $i=0,\ldots,k$.  That is, stage 0 begins on the first probe, and stage $i$ begins with the $m^{2(i-1)}+1$'th probe for all $i=1,\ldots,k$.
% The probing starts at stage $0$, when
% nothing has been probed and it is only known that the patience is at least 1. If
% the first probe fails and the patience has not been exhausted, then we enter
% stage $1$ where it is known that the patience is at least $m^{2}$. Similarly,
% after $m^{2}$ failed probes, if the patience has not been exhausted, then we
% enter stage $2$ where it is known that the patience is at least $m^{4}$. The
% final stage is stage $k-1$, where it is known that the patience has the maximum
% possible value of $m^{2(k-1)}$.
We will let $V(i)$ denote the expected reward
received from stage $i$ onward, conditioned on reaching stage $i$ (with the online vertex still unmatched). Note that the
expected reward of the offline algorithm is then given by $V(0)$.

We proceed by backward induction. First, consider the case where the algorithm
has reached the final stage $i=k$. This occurs after $m^{2(k-1)}$ unsuccessful probes,
so that there are $m^{2(k-1)}(m^{2} - 1)$ remaining probes before
the patience is exhausted. In this case, the expected reward is at most
$m^{k}$, since the maximum weight of any edge is $m^{k}$.

Our inductive hypothesis is $V(i) \le m^{i}+(k-i)m^{i-1}$, which we already demonstrated holds in the base case $i=k$.
Now, consider an arbitrary stage $i < k$.  At the beginning of stage $i$, the
algorithm has made $m^{2(i-1)}$ probes, and the patience is known to be at least
$m^{2i}$. The stage ends after $m^{2i} - m^{2(i-1)}$ unsuccessful probes (it may
end sooner if any probe is successful), at which point the algorithm either
learns that the patience is exactly $m^{2i}$ (and can make no more probes), or
learns that the patience is at least $m^{2(i+1)}$ (and enters stage $i+1$).

We make the following observations about the optimal online algorithm at stage $i$.  The inductive hypothesis implies that $V(i+1)\le m^{i+1}+(k-(i+1))m^i$.  Moreover, conditional on the patience being at least $m^{2i}$, the probability of the patience being at least $m^{2(i+1)}$ (i.e., the next stage existing) is $1/m$.  Therefore, $V(i)$ cannot be greater than in a hypothetical universe where the online algorithm can make $m^{2i} - m^{2(i-1)}$ probes, and if none of them are successful, then collects a deterministic reward of $\frac{1}{m}(m^{i+1}+(k-(i+1))m^i)=m^{i}+(k-(i+1))m^{i-1}$.

In this hypothetical universe, the online algorithm should never probe an edge with weight $m^{i'}$ with $i'\le i$, since it would collect a better reward from waiting for stage $i+1$.  On the other hand, if it probes only edges with weight $m^{i'}$ with $i'>i$ during its $m^{2i} - m^{2(i-1)}$ probes, then $V(i)$ can be no greater than
\begin{align} \label{eqn:hypUni}
m^{2i}m^{-(i+1)}+m^{i}+(k-(i+1))m^{i-1}
% &=m^{i-1}+m^{i}+(k-(i+1))m^{i-1}
% \\ &=m^{i}+(k-i)m^{i-1}
\end{align}
since the online algorithm has $m^{2i} - m^{2(i-1)}\le m^{2i}$ probes and the expected value of each probe can be at most $m^{-2i'}m^{i'}=m^{-i'}\le m^{-(i+1)}$, with the final term $m^{i}+(k-(i+1))m^{i-1}$ representing the online algorithm's reward if none of the probes are successful (expression~\eqref{eqn:hypUni} assumes this reward is always collected, which is clearly a valid upper bound).  Expression~\eqref{eqn:hypUni} equals $m^{i}+(k-i)m^{i-1}$, completing the induction.

Therefore, we have established that $V(0)\le 1+k/m$.  The ratio between the expected earnings of the best online algorithm relative to the offline (which knows the patience realization in advance) is at most
\begin{align*}
\frac{1+k/m}{(1-1/e)(k+1-k/m)}
\end{align*}
which approaches $\frac{1}{(1-1/e)(k+1)}$ as $m\to\infty$, completing the proof.
\Halmos\end{proof}

\section{Relation to the Online Assortment Problem} \label{sec:relation_to_assortment}

The work of~\cite{GNR} considers a model in which online vertices represent customers and offline vertices represent products, and a merchant wishes to offer products to consumers so as to maximize profit. This setting differs from ours in that the merchant offers a collection of several products all at once. Then, the customer chooses to either purchase some product (or multiple products at once) based on products offered or purchase nothing. By contrast, in our model the algorithm (the "merchant" in our setting) attempts one match at a time, stopping when a successful match occurs or the number of unsuccessful attempts equals the patience constraint.

In the setting of~\cite{GNR}, each customer $v$ has a "general choice model" $\phi_{v}(S, u)$ that specifies the probability that customer $v$ purchases item $u$ when offered the set $S$ of items. More generally, since the model considers that $v$ may purchase more that one item, $\phi_{v}(S, S')$ is used to denote the probability that $v$ will purchase exactly the items $S'$ when offered $S$ (and then $\phi_{v}(S, u)$ is defined to be $\sum_{S' : u\in S'}\phi_{v}(S, S')$). It is assumed that the customer will only purchase products that were offered as part of the assortment $S$ (that is, $\phi_{v}(S, S') = 0$ if $S'\not\subseteq S$). 

The algorithm they propose for their model can be viewed as a greedy algorithm which presents an online-arriving customer $v$ with the set $S$ that maximizes the expected profit of the items $v$ purchases. Doing so would guarantee a competitive ratio of at least $0.5$, though this maximization step is not necessarily solvable in polynomial time for arbitrary choice models.
%(they present only a specific family of choice models for which this step can be solved in polynomial time).

Their results do not immediately extend to our setting, as their stochastic model is somewhat different. Extending their results to our setting requires a reduction from our sequential probing with the probe-commit model to this \emph{all-at-once} model by construction of appropriate choice models $\phi$. Further, such a reduction would not necessarily yield a polynomial-time result without also designing an algorithm for solving the aforementioned maximization in polynomial time.

One contribution of the present work is \Cref{alg:dpgreedy}, which indeed can be viewed as
greedily maximizing the expected weight (or profit) of $v$'s match (or
purchase). However, without also constructing a reduction from our sequential
probing model to this all-at-once model, the result of~\cite{GNR} does not extend to
give a competitive ratio of $0.5$ for our problem. Rather, in the present work,
we present a clean, self-contained analysis of \Cref{alg:dpgreedy} to achieve a competitive
ratio of $0.5$ for our problem without relying on the results of~\cite{GNR} or the
similar framework of~\cite{cheung2022inventory} based on "actions".

\section{Proof of \Cref{prop:potentialBased} from \citet{CSL16}} \label{sec:potential_based}

The potential-based framework of \citet{CSL16} finds for each $v\in V$ a policy $\pi$ that is $\kappa$-approximate to the problem of $\max_{\pi\in\cP}\sum_{u\in U}(w_{uv}-\alpha_u)p_{uv}(\pi)$, updates the dual variables $\alpha_u$ accordingly, and then repeats iteratively.
The analysis uses a potential function motivated by the multiplicative weights method \cite{arora2012multiplicative}.
The final solution returned is the average of all the policies found over the iterations, and the total number of iterations needed to average into a feasible and $(\kappa-\epsilon)$-approximate LP solution is polynomial in $1/\epsilon$.

Although it is stated in the context of assortment optimization, Theorem~3.4 from \cite{CSL16} (and its generalization in Appendix~G of~\cite{CSL16}) directly implies our \Cref{prop:potentialBased}, after the following observations about the correspondence with our context of online matching with patience.
In assortment optimization, there are an exponential number of "assortments" $S$ (unordered subsets) of offline vertices that could be offered to an online vertex of type $v$, after which there is a probability $p_{uv}(S)$ of each $u\in S$ being matched.
The only ingredient needed in \cite{CSL16} is that given any set of (possibly negative) adjusted weights $w'_{uv}$ for an online type $v$, an assortment which solves $\max_S\sum_{u\in S}w'_{uv}p_{uv}(S)$ within a factor of $\kappa$ can be found in polynomial time.
The framework is unchanged if in our setting, we consider policies $\pi$ which are \emph{ordered} subsets of offline vertices instead.
Moreover, our black boxes for finding optimal or approximately optimal policies $\pi$ can ignore any negative adjusted weights, since it is never beneficial to probe a negative edge.
Finally, the general packing constraints and flexible products allowed in \cite{CSL16} are generalizations which can be ignored, and as a result the statement of Theorem~3.4 from \cite{CSL16} implies our \Cref{prop:potentialBased}.

\end{APPENDICES}

%\theendnotes

% References here (outcomment the appropriate case)

% CASE 1: BiBTeX used to constantly update the references
%   (while the paper is being written).
%\bibliographystyle{informs2014} % outcomment this and next line in Case 1
%\bibliography{bibliography} % if more than one, comma separated

% CASE 2: BiBTeX used to generate mypaper.bbl (to be further fine tuned)

%If you don't use BiBTex, you can manually itemize references as shown below.

%%%%%%%%%%%%%%%%%
\end{document}